\documentclass[journal]{IEEEtran}
\usepackage{amsmath, amssymb, amsthm}
\usepackage{algorithmic}
\usepackage{algorithm}
\usepackage{array}
\usepackage[caption=false,font=small,labelfont=sf,textfont=sf]{subfig}

\usepackage{textcomp}
\usepackage{stfloats}
\usepackage{url}
\usepackage{verbatim}
\usepackage{cite}
\usepackage{graphicx}

\usepackage{multirow}
\usepackage{booktabs}
\usepackage{kotex}
\usepackage{color}

\newcommand{\Ncp}{N_{\text{cp}}}
\newcommand{\Ns}{N_{\text{s}}}
\newcommand{\NRB}{N_{\text{RB}}}
\newcommand{\NRE}{N_{\text{RE}}}
\newcommand{\NTBS}{N_{\text{TBS}}}
\newcommand{\Ninfo}{N_{\text{info}}}

\newcommand{\Nt}{N_{\text{T}}}

\newcommand{\Nr}{N_{\text{R}}}

\newcommand{\snreesm}{\text{SNR}_{\text{EESM}}}

\newcommand{\diag}{\text{diag}}

\newcommand{\norm}[1]{\left\lVert#1\right\rVert}

\DeclareMathOperator*{\argmax}{arg\,max}

\newcommand{\figref}[1]{Fig.~\ref{#1}}
\newcommand{\tabref}[1]{Table~\ref{#1}}
\newcommand{\secref}[1]{Section~\ref{#1}}
\newcommand{\algref}[1]{Algorithm~\ref{#1}}

\def\bb0{{\mathbb{0}}}

\def\ba{{\mathbf{a}}}
\def\bb{{\mathbf{b}}}

\def\bp{{\mathbf{p}}}

\def\bs{{\mathbf{s}}}

\def\bv{{\mathbf{v}}}

\def\bx{{\mathbf{x}}}
\def\by{{\mathbf{y}}}

\def\b0{{\mathbf{0}}}

\def\bA{{\mathbf{A}}}

\def\bF{{\mathbf{F}}}

\def\bH{{\mathbf{H}}}
\def\bI{{\mathbf{I}}}

\def\bU{{\mathbf{U}}}
\def\bV{{\mathbf{V}}}

\def\bbC{{\mathbb{C}}}

\def\bbE{{\mathbb{E}}}

\def\bbR{{\mathbb{R}}}

\def\cA{\mathcal{A}}

\def\cI{\mathcal{I}}

\def\cM{\mathcal{M}}
\def\cN{\mathcal{N}}

\def\cR{\mathcal{R}}

\def\sf0{{\mathsf{0}}}

\def\rm0{{\mathrm{0}}}

\newtheorem{theorem}{Theorem}

\theoremstyle{definition}
\newtheorem{remark}{Remark}

\begin{document}
\bstctlcite{IEEEexample:BSTcontrol}

\title{Energy Efficiency Optimization in Distributed MIMO vRAN via Cross-Layer Link Abstraction}

\author{Jaebum Park,~\IEEEmembership{Graduate Student Member,~IEEE},~Chan-Byoung Chae,~\IEEEmembership{Fellow,~IEEE}, and\\ Robert W. Heath Jr.,~\IEEEmembership{Fellow,~IEEE}

\thanks{This material is based upon work supported in part by the National Science Foundation under Grant No. NSF ECCS-2414678 and in part by the Army Research Office under Grant W911NF2410107, and in part by the Korean Government under Grants (RS-2026-25525793, RS-2024-00397216).}
\thanks{J. Park and R. W. Heath, Jr. are with the Department of Electrical and Computer Engineering, University of California at San Diego, La Jolla, CA 92093, USA (emails: \{jp261, rwheathjr\}@ucsd.edu).} 
\thanks{C.-B. Chae is with the School of Integrated Technology, Yonsei University, Seoul 03722, South Korea (email: cbchae@yonsei.ac.kr).}}

\maketitle


\begin{abstract}
Virtualized radio access networks (vRAN) run the compute-intensive multiple-input multiple-output (MIMO) baseband as software on shared servers, which makes energy efficiency (EE) a primary design objective. Distributed MIMO vRAN consumes power across virtualized distributed unit (vDU) baseband, fronthaul transport, and per-radio-unit operation. We build a power model that resolves these three components. We then develop a framework that jointly selects modulation, transmission rank, and per-subcarrier power to maximize system EE. Exponential effective SNR mapping induces a convex per-subcarrier power constraint, which yields a convex power minimization problem with a closed-form waterfilling-like solution. We show that radio frequency-only models underestimate the spectral efficiency range where single-input multiple-output (SIMO) transmission saves power, and our power model extends this range by 24\%. We further extend the framework to a traffic-aware setting with realistic user trajectories from the multi-agent transport simulator. We propose a traffic-aware strategy that switches each radio unit among MIMO, SIMO, and sleep modes based on demand. Simulation results over 3GPP NR compliant fading channels show that, after a one-time offline calibration, the framework predicts link performance without further link-level simulation. The proposed framework achieves higher average EE than a traffic-agnostic always-on MIMO baseline, while maintaining comparable throughput at peak hours.
\end{abstract}


\section{Introduction}

\IEEEPARstart{M}{ultiple}-input multiple-output (MIMO) virtualized radio access networks (vRAN) move the compute-intensive MIMO baseband, including channel estimation, precoding, and detection, from dedicated hardware to software on shared general-purpose servers~\cite{ParkEtAlAcceleratingVRANORAN2026}. Functional splits place each function where it runs most efficiently, and centralizing baseband across cells makes coordinated transmission such as network MIMO practical. Pooling compute on shared servers scales with load and lowers cost, while open interfaces reduce vendor lock-in and let operators deploy new MIMO methods in software~\cite{oranWG4CUS}. The MIMO baseband processing is the most compute-intensive part of the RAN. Its energy grows with the number of antennas, users, and offered load. Energy efficiency (EE) in MIMO vRAN is therefore a primary design objective. RAN already dominates mobile network energy consumption, and 5G NR densification amplifies this share~\cite{GreenFutureNetworks2021}.

Optimizing EE in vRAN demands a new modeling approach, because the architecture spreads power across components that traditional models never captured. vRAN disaggregates the base station into a vDU and geographically separated radio units (RUs) connected via fronthaul links, which introduces vDU baseband, fronthaul transport, and per-RU power components that radio frequency (RF)-only EE frameworks do not model. Cell-free distributed MIMO (D-MIMO) compounds this challenge, since coherent joint transmission across distributed RUs~\cite{NgoEtAlCellFreeMassiveMIMO2017a} scales both vDU load and fronthaul transport with the number of active RUs. Energy-efficient D-MIMO vRAN operation therefore requires a power model that resolves these components, an EE optimization that jointly selects modulation, transmission rank, and power under a closed-form structure, and a control strategy that adapts RU operation to spatio-temporal traffic. In this paper, we propose a new model for EE, design an algorithm that adapts the MIMO strategy to maximize EE, and generalize the algorithm to allow partial sleep of the hardware.

\subsection{Prior work}

Prior work on modeling EE in MIMO systems has built detailed power models for base stations and cellular networks. Early base-station power models extended to MCS-dependent characterizations for 5G systems~\cite{AuerEtAlHowMuchEnergy2011, PiovesanEtAlMachineLearningAnalytical2022}. Empirical studies on production vRAN servers~\cite{KaliaEtAlEnergyEfficient5G2025} reported that CPU power varies with sub-second workload fluctuations. At the system level, several vRAN and Open-RAN models decomposed power into DU, RU, and transport components~\cite{NgoEtAlTotalEnergyEfficiency2018, SinghEtAlEnergyEfficientOrchestrationMetroScale2021, DemirEtAlCellfreeMassiveMIMO2024, TopalEtAlUnlockingEnergySavingPotential2026}. These models addressed metro-scale orchestration with vDU placement and RU activation~\cite{SinghEtAlEnergyEfficientOrchestrationMetroScale2021}, cloud-fronthaul-radio joint orchestration for cell-free massive MIMO~\cite{DemirEtAlCellfreeMassiveMIMO2024}, and wireless-fronthaul extensions~\cite{TopalEtAlUnlockingEnergySavingPotential2026}. Unfortunately, none of these models account for the vDU baseband, fronthaul transport, and per-RU components of a D-MIMO vRAN simultaneously.

Prior work on adaptation in MIMO systems has optimized EE by allocating rate and power across users and subcarriers. Most work relied on Shannon capacity formulations~\cite{IshedenEtAlFrameworkLinkLevelEnergy2012, ShenYuFractionalProgrammingCommunication2018, ShenYuFractionalProgrammingCommunication2018a}, demonstrating EE gains in single-cell OFDMA systems~\cite{NgEtAlEnergyEfficientResourceAllocation2012}, frequency-selective OFDM link adaptation with circuit power~\cite{MiaoEtAlEnergyefficientLinkAdaptation2010}, multi-cell NOMA networks~\cite{FangEtAlEnergyEfficientResourceAllocation2016}, and multi-user MIMO systems~\cite{BjornsonEtAlOptimalDesignEnergyefficient2015}. Subsequent work on hardware-aware adaptive PHY~\cite{GastEtAlHardwareawareEnergyEfficiency2024} and cross-layer adaptive MIMO~\cite{KimEtAlCrosslayerApproachEnergy2009, ChaeEtAlAdaptiveMIMOTransmission2010} similarly adopted Shannon capacity in single-RU settings. Shannon capacity provides an idealized continuous-rate upper bound. Deployed 3GPP systems transmit over discrete MCS levels under a target BLER, so the achievable rate falls below this bound~\cite{FengEtAlSurveyEnergyefficientWireless2013a}. Unfortunately, an EE optimization built on Shannon capacity does not capture these MCS and BLER constraints. EE optimization for D-MIMO vRAN therefore needs a practical link abstraction that maps MCS and rank selections to throughput under a target BLER. Effective SNR mapping is the dominant such abstraction, with capacity, mutual-information, and exponential variants~\cite{HeathLozanoFoundationsMIMOCommunications2018, WanEtAlFadinginsensitivePerformanceMetric2006}. Among these, exponential effective SNR mapping (EESM) has been widely adopted for link adaptation and proven analytically tractable in multi-cell OFDMA settings~\cite{BlankenshipEtAlLinkErrorPrediction2004a, BrueninghausEtAlLinkPerformanceModels2005, FrancisMehtaEESMBasedLinkAdaptation2014}. EESM is itself a mutual-information approximation, but its exponential form admits a closed-form inverse. Other mutual-information-based and ML-based mappings improve prediction accuracy. Their compression functions, however, lack a closed-form inverse or rely on data-driven mappings, which precludes the tractable optimization that a closed-form inverse enables. EESM therefore offers a closed-form inverse that remains underexploited for cross-layer EE optimization.

In addition to adapting the mode of operation, the hardware can also enter sleep mode. Prior work on sleep mode strategies adapts RU activation to mobile traffic variation. Surveys~\cite{WuEtAlEnergyEfficientBaseStationsSleepMode2015} and analyses~\cite{FengEtAlBaseStationONOFF2017} established base station on-off switching as viable. The same idea extends to access-point sleep mode methods for cell-free massive MIMO under non-uniform traffic~\cite{Garcia-MoralesEtAlEnergyEfficientAccessPointSleepMode2020}. Recent Open-RAN-native approaches use deep-reinforcement-learning-based xApps on the RAN intelligent controller (RIC) for dynamic RU activation~\cite{BordinEtAlDesignEvaluationDeep2025}. Unfortunately, these methods lack the interpretability of closed-form solutions and do not jointly optimize the operation mode with PHY-layer parameters such as MCS, rank, and power.

\subsection{Contributions}

In this paper, we develop a cross-layer EE optimization framework for D-MIMO vRAN. First, we build a power model that resolves vDU baseband, fronthaul transport, and per-RU power. We then exploit the exponential structure of EESM to reformulate EE maximization as convex power minimization, which yields a waterfilling-like closed-form allocation. Finally, we extend this closed-form structure to a traffic-aware strategy that switches each RU among MIMO, SIMO, and sleep modes. The main contributions are as follows.

\begin{itemize}

\item We develop a D-MIMO vRAN power model that captures vDU processing, fronthaul, and distributed RU operations under dynamic traffic. Unlike our prior cross-layer framework~\cite{KimEtAlCrosslayerApproachEnergy2009}, which neglects baseband and fronthaul overhead, the model resolves all three components jointly.

\item We reformulate EE maximization as a convex power minimization problem under 3GPP-compliant transport block size (TBS). We solve the formulation by jointly optimizing MCS, rank, physical resource block (PRB) allocation, and power. The exponential structure of EESM yields closed-form Karush-Kuhn-Tucker (KKT) conditions and thus waterfilling-like solutions. An offline additive white Gaussian noise (AWGN)-calibrated BLER curve removes per-configuration link-level simulation from the optimization loop. We validate that a single calibrated $\beta_m$ per MCS matches the actual BLER over TDL-C channels.

\item We propose a traffic-aware RU mode selection strategy that dynamically switches each RU between MIMO, reduced-rank SIMO, and sleep mode based on traffic demand. Unlike prior sleep mode approaches~\cite{FengEtAlBaseStationONOFF2017, Garcia-MoralesEtAlEnergyEfficientAccessPointSleepMode2020}, the strategy jointly optimizes the operation mode and PHY parameters. To evaluate this strategy under realistic load, we build a spatio-temporal traffic model that combines multi-agent transport simulator (MATSim) user mobility with an empirical APN-C demand profile, since no standardized model captures joint spatial and temporal variation. The strategy achieves a 32\% average EE gain over an always-on MIMO baseline while maintaining quality-of-service.

\end{itemize}

A preliminary version of the traffic-aware mode selection appeared in~\cite{ParkEtAlSwitchDFTAdaptiveWaveform2026}.  We extend this preliminary version by incorporating vDU and fronthaul power, deriving closed-form EESM-based power allocation, and evaluating performance under realistic 3D spatio-temporal traffic.

This paper is organized as follows. In \secref{sec:system_model}, we describe the D-MIMO vRAN architecture and develop an OFDM signal model that captures distributed RU operations. In \secref{sec:power_model}, we build the power model that resolves vDU baseband, fronthaul transport, and per-RU power. In \secref{sec:EE_DMIMO}, we develop the cross-layer EE optimization framework under fixed traffic load and derive the closed-form power allocation. In \secref{sec:traffic-aware-d-mimo}, we extend the framework to spatio-temporal traffic dynamics and propose a traffic-aware RU mode selection strategy. In \secref{sec:results}, we validate the EESM calibration and present EE results over 3GPP NR-compliant channels and MATSim-generated traffic. Lastly, we conclude in \secref{sec:conclusion}.

\textit{Notation}: We denote a bold lowercase $\ba$ as a column vector and a bold uppercase $\bA$ as a matrix. We use $\bA^H$ and $\bA^T$ for the conjugate transpose and transpose of $\bA$. We use $|\cdot|$ to denote the absolute value of a scalar or the cardinality of a set, $\norm{\cdot}_F$ the Frobenius norm, and $\bbE[\cdot]$ the expectation. We use $\mathbf{1}$ and $\mathbf{0}$ to denote the all-ones and all-zeros column vectors, and $\bI$ the identity matrix, all of appropriate dimension. The set of real numbers is $\bbR$ and complex numbers is $\bbC$.


\section{System model}
\label{sec:system_model}

We lay out a single-user D-MIMO vRAN downlink signal model with OFDM transmission. We focus on the downlink because rate adaptation and optimization are more challenging there, and power consumption is generally higher on the downlink than on the uplink. We discuss extensions to multi-user settings in \secref{sec:conclusion}. As shown in~\figref{fig:vRAN_dMIMO}, the vDU implements baseband processing on general-purpose servers~\cite{ParkEtAlAcceleratingVRANORAN2026}. The vDU connects to distributed RUs via an Ethernet switch over dedicated fiber optic links, each supporting enhanced Common Public Radio Interface (eCPRI)-based fronthaul transport. Each RU houses the RF chains and power amplifiers that drive its transmit antennas. This disaggregation separates the system into a vDU, fronthaul links, and distributed RUs. We first explain the system dimensionality. We then develop the distributed MIMO channel model and the received signal model, and we make several assumptions to use them. Our goals are to compute spectral efficiency, to develop an algorithm that approaches that spectral efficiency, and then to compute energy efficiency for this system. To support these computations, we further specify the 3GPP NR frame structure and TBS, and conclude with the fronthaul rate requirements.

\begin{figure}[t]
    \centering
    \includegraphics[width=0.45\textwidth]{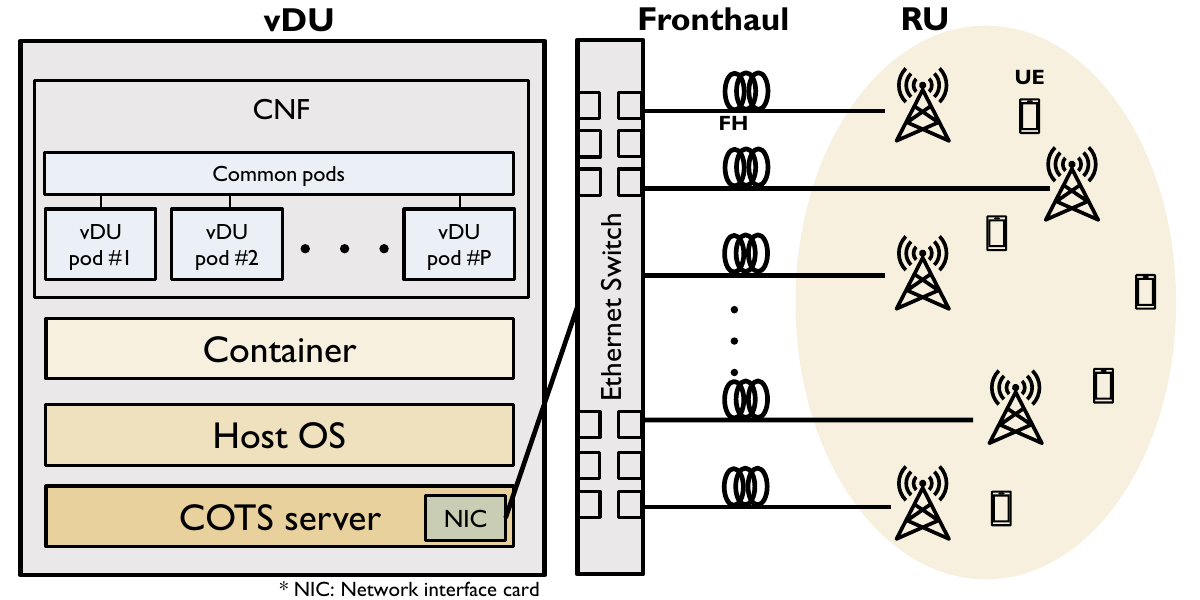}
    \caption{D-MIMO vRAN architecture. The vDU runs containerized network functions (CNFs) on a general-purpose server~\cite{ParkEtAlAcceleratingVRANORAN2026}. CNFs are software modules that share server hardware and scale resources on demand. The vDU connects to distributed RUs via fronthaul optical links using eCPRI under Open-RAN functional split 7-2x. This disaggregation introduces three power components, namely vDU baseband processing, fronthaul transport, and per-RU power.}
    \label{fig:vRAN_dMIMO}
\end{figure}

\subsection{System dimensionality}
\label{subsec:dimensionality}
A centralized vDU coordinates $L$ distributed RUs connected via fronthaul links. The $\ell$-th RU is equipped with $\Nt$ transmit antennas, and the UE is equipped with $\Nr$ receive antennas, yielding an $L\Nt \times \Nr$ MIMO configuration. Throughout this paper, index $k \in \{0, 1, \ldots,K-1\}$ refers to a specific subcarrier, $t$ to a time slot, and $\ell \in \{1,2,\ldots,L\}$ to a specific RU. The system supports up to $\Ns = \min(\Nr, L\Nt)$ spatial streams. The actual transmission rank $r \in \{1, \ldots, \Ns\}$ is selected adaptively per slot.

\begin{figure*}[!t]
  \centering
  \subfloat[]{\includegraphics[width=0.6\textwidth]{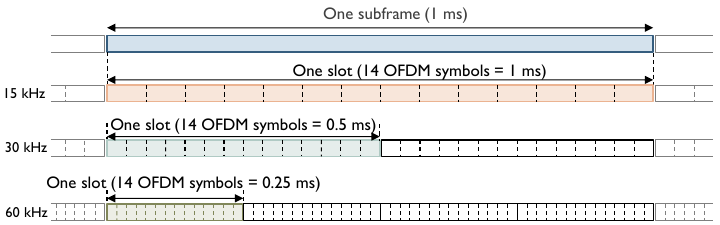} \label{fig:slot}} 
  \hfill
  \subfloat[]{\includegraphics[width=0.38\textwidth]{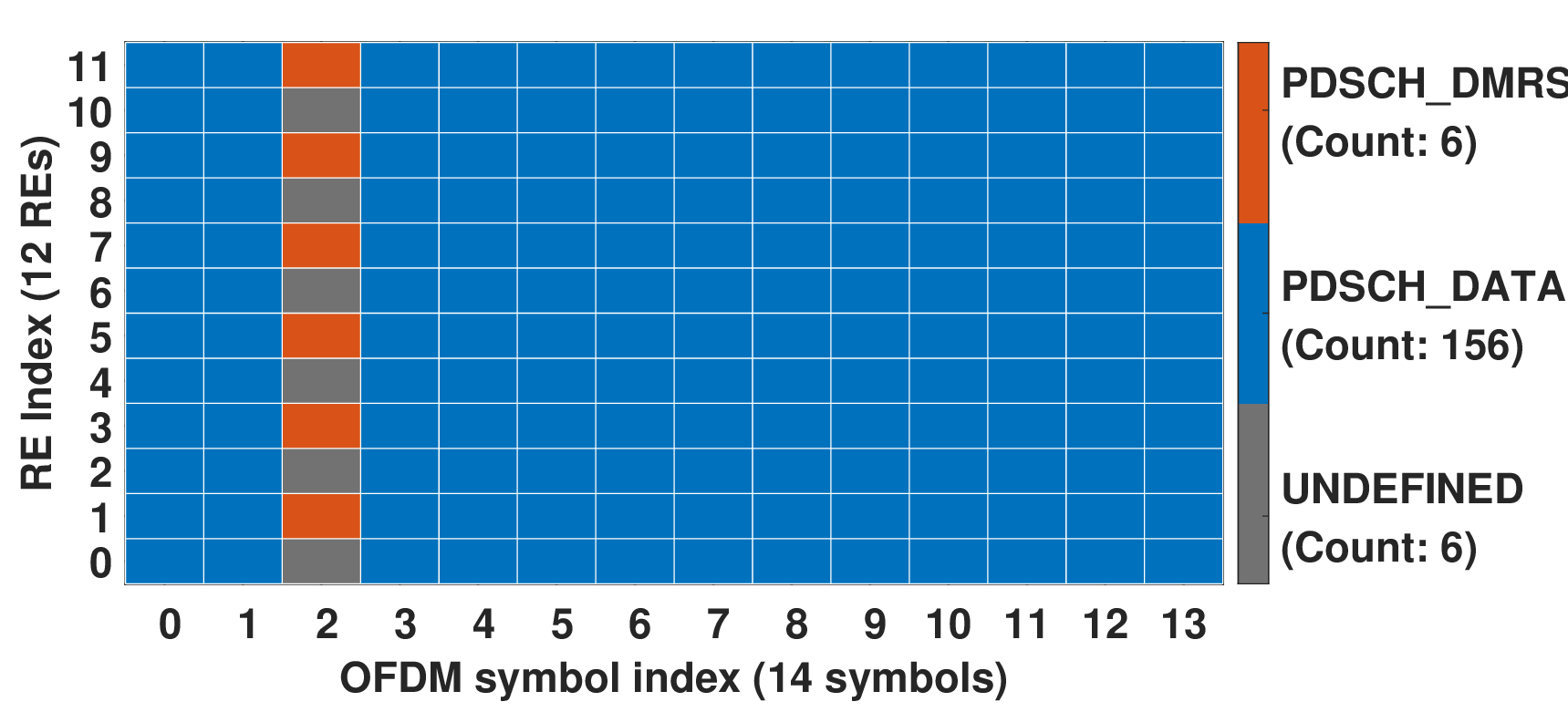} \label{fig:RB}} 
  \caption{3GPP NR physical layer resource structure. (a) Slot duration versus subcarrier spacing for $\mu \in \{0, 1, 2\}$. (b) Resource grid for one PRB (single layer) showing PDSCH data and DMRS symbols. DMRS occupies 6 of 168 REs, leaving 156 data REs.}
  \label{fig:NR_structure}
\end{figure*}

\subsection{Distributed MIMO channel model}
\label{subsec:channel}

Let $G_{\ell}$ denote the large-scale fading coefficient that accounts for path loss~\cite{HeathLozanoFoundationsMIMOCommunications2018}, and let $\tilde{\bH}_{\ell}[k] \in \bbC^{\Nr \times \Nt}$ denote the small-scale fading matrix normalized such that $\bbE[\|\bH_{\ell}[k]\|_F^2] = G_{\ell} \Nr \Nt$. The frequency-domain channel matrix at subcarrier $k$ between RU $\ell$ and the UE is
\begin{equation}
\bH_{\ell}[k] = \sqrt{G_{\ell}}\, \tilde{\bH}_{\ell}[k].
\label{eq:complete_channel}
\end{equation}
This decomposition separates the per-RU path loss $G_\ell$, set by the RU-UE distance, from the small-scale fading $\tilde{\bH}_\ell[k]$. The distinct $G_\ell$ per RU reflects the different geographic locations of the RUs~\cite{3gppTR38.901}.

\subsection{Received signal model}
\label{subsec:signal}

The vDU computes a precoding matrix $\bF_{\ell}[k] \in \bbC^{\Nt \times r}$ for each RU. The data symbol vector $\bs[k] \in \bbC^{r \times 1}$ satisfies $\bbE[\bs[k]\bs[k]^H] = \bI$. The transmit power at RU $\ell$ is $P_{\ell}[k] \geq 0$. The transmitted signal from RU $\ell$ is $\bx_{\ell}[k] \in \bbC^{\Nt \times 1}$, given by
\begin{equation}
    \bx_{\ell}[k] = \sqrt{P_{\ell}[k]}\,\bF_{\ell}[k]\,\bs[k].
    \label{eq:transmitted_signal_per_ru}
\end{equation}
Let $\bv[k] \sim \cN_{\bbC}(\b0,\, N_0\Delta f\,\bI_{\Nr})$ denote independent complex Gaussian noise across subcarriers, where $N_0$ is the noise power spectral density and $N_0\Delta f$ is the noise power per subcarrier. The UE receives the aggregate signal from $L_\text{act}$ active RUs as
\begin{align}
    \by[k] &= \sum_{\ell=1}^{L_\text{act}} \bH_{\ell}[k]\, \bx_{\ell}[k] + \bv[k] \nonumber \\
          &= \sum_{\ell=1}^{L_\text{act}} \sqrt{P_{\ell}[k]} \bH_{\ell}[k] \,\bF_{\ell}[k]\,\bs[k] + \bv[k].
\end{align}
The single-user model carries no inter-user interference, since the vDU serves one UE over the allocated resources. The RUs also operate cell-free under a common physical cell identity, so the UE perceives a single-cell downlink. The signal model therefore neglects out-of-system interference and represents the link as noise-limited. The vDU jointly designs the precoders $\bF_\ell[k]$ across the distributed RUs, and the UE combines the received signal, so the two operations together diagonalize the aggregate channel into parallel interference-free streams.

The vDU implements centralized joint processing across distributed RUs~\cite{BjornsonSanguinettiMakingCellFreeMassive2020}, estimating the concatenated channel matrix $\bH[k] = [\bH_1[k]\;\bH_2[k]\;\cdots\;\bH_L[k]] \in \bbC^{\Nr \times L\Nt}$ from uplink reference signals. We assume perfect CSI at the vDU via time-division duplexing (TDD) reciprocity and perfect synchronization across the distributed RUs. Let $\bU[k] \in \bbC^{\Nr \times \Nr}$ and $\bV[k] \in \bbC^{L\Nt \times L\Nt}$ be unitary matrices, and let $\mathbf{\Sigma}[k] \in \bbR^{\Nr \times L\Nt}$ contain the ordered singular values $\sigma_1[k] \geq \sigma_2[k] \geq \cdots \geq \sigma_{\Ns}[k] \geq 0$. From the singular value decomposition (SVD) of $\bH[k] = \bU[k]\,\mathbf{\Sigma}[k]\,\bV[k]^H$, the vDU builds a joint precoder from the $r$ dominant right singular vectors, extracting $\bF_\ell[k] = \bV[k]_{(\ell-1)\Nt+1:\ell\Nt,\,1:r} \in \bbC^{\Nt \times r}$ for RU~$\ell$. The vDU distributes $\bF_\ell[k]$ to each RU via the fronthaul. Since $\bV[k]$ is unitary, the precoder blocks satisfy $\sum_{\ell=1}^{L} \norm{\bF_\ell[k]}_F^2 = r$. The precoder thus fixes only the spatial mapping, while each RU sets its transmit power independently through $P_\ell[k]$ under its power amplifier budget.

Let $\bV_r[k]$ and $\bU_r[k]$ collect the first $r$ columns of $\bV[k]$ and $\bU[k]$. The stacked precoder $\bF[k] = [\bF_1[k]^T \cdots \bF_L[k]^T]^T$ equals $\bV_r[k]$, so the effective channel reduces to
\begin{equation}
    \bH[k]\,\bF[k] = \bU[k]\,\mathbf{\Sigma}[k]\,\bV[k]^H \bV_r[k] = \bU_r[k]\,\mathbf{\Sigma}_r[k],
\label{eq:effective_channel}
\end{equation}
where $\mathbf{\Sigma}_r[k] = \diag(\sigma_1[k], \ldots, \sigma_r[k])$ collects the $r$ largest singular values. The UE combines the received signal with $\bU_r[k]^H$. The vDU assigns a common per-subcarrier power $P[k]$ shared equally across the $r$ streams. Using \eqref{eq:effective_channel} and $\bU_r[k]^H\bU_r[k] = \bI_r$, the combined signal becomes
\begin{equation}
    \bU_r[k]^H\by[k] = \sqrt{\tfrac{P[k]}{r}}\,\mathbf{\Sigma}_r[k]\,\bs[k] + \bU_r[k]^H\bv[k].
\label{eq:combined_signal}
\end{equation}
Since $\bU[k]$ is unitary, the noise $\bU_r[k]^H\bv[k]$ keeps the per-subcarrier power $N_0\Delta f$. The per-stream received SNR is
\begin{equation}
    \gamma_q[k] = \frac{\sigma_q^2[k]\,P[k]}{r\,N_0\Delta f}, \quad q = 1, \ldots, r.
\label{eq:per_stream_snr}
\end{equation}
The diagonalization follows from SVD precoding $\bV_r[k]$ at the transmitter and combining $\bU_r[k]^H$ at the receiver, as shown in \eqref{eq:effective_channel} and \eqref{eq:combined_signal}. The combined operation removes inter-stream interference, so each stream sees an independent channel with SNR $\gamma_q[k]$.

\subsection{3GPP NR frame structure and MAC/PHY procedures}

We develop link adaptation strategies on a 5G NR frame structure, so we first introduce the key concepts, terminology, and parameters that the optimization in Sections~\ref{sec:EE_DMIMO} and~\ref{sec:traffic-aware-d-mimo} uses. The medium access control (MAC) scheduler selects MCS, PRB allocation, and rank based on UE channel feedback, which together determine the TBS. The PHY layer then maps the transport block onto the OFDM resource grid defined by the NR frame structure. The subcarrier spacing $\Delta f$, slot duration $T_{\text{slot}}$, cyclic prefix samples $\Ncp$, data resource elements $\NRE$, and MCS parameters $Q_m$ and $R_m$ appear directly in the effective SNR mapping, throughput, and TBS expressions. We adopt the 3GPP NR frame structure with flexible numerology $\mu$, which sets the subcarrier spacing $\Delta f = 2^{\mu} \times 15$ kHz and the slot duration $T_{\text{slot}} = 2^{-\mu}$ ms. Each OFDM symbol carries $\Ncp$ cyclic prefix samples. Figure~\ref{fig:NR_structure}(a) illustrates the resulting time-frequency grid.

The TBS calculation follows 3GPP TS 38.214~\cite{3gppTS38.214}. We denote $\NRB$ as the maximum number of PRBs in the system bandwidth and $n$ as the number allocated for transmission, where $n \in \{1, \ldots, \NRB\}$. Let $N_{\text{symb}}$ denote the OFDM symbols per slot allocated to the physical downlink shared channel (PDSCH), $N_{\text{DMRS}}^{\text{PRB}}$ the demodulation reference signal (DMRS) overhead per PRB, and $N_{\text{oh}}^{\text{PRB}}$ additional overhead from phase-tracking reference signals or reserved resources. The total number of data REs per PRB is $\NRE^{\text{PRB}}$, and the total data REs for PDSCH transmission is
\begin{equation}
    \NRE = n\NRE^{\text{PRB}} =  n \left( 12  N_{\text{symb}} - N_{\text{DMRS}}^{\text{PRB}}  - N_{\text{oh}}^{\text{PRB}}\right).
\label{eq:n_re}
\end{equation}
\tabref{tab:mcs} lists the MCS parameters for PDSCH used in this work. Four representative MCS indices spanning quadrature phase-shift keying (QPSK, $m=0$) to 256-quadrature amplitude modulation (256-QAM, $m=3$) are selected to cover the full modulation range. We denote $Q_m$ as the modulation order and $R_m$ as the code rate for MCS index $m$. The total number of information bits is
\begin{equation}
\Ninfo = \NRE \cdot R_m \cdot Q_m \cdot r.
\label{eq:ninfo}
\end{equation}
For analytical tractability, we omit the 3GPP TBS quantization for low-density parity-check (LDPC) compatibility and treat $\Ninfo$ directly as the transport block size $\NTBS$, flooring to the nearest bit. For numerical evaluation, we use $\NRE^{\text{PRB}} = 156$ REs per PRB based on the DMRS configuration shown in Fig.~\ref{fig:NR_structure}(b). The structure of these expressions does not depend on the specific 3GPP NR values. It remains valid for any standard that provides a finite set of MCS, PRB, and rank options, so the framework extends to 6G and beyond. We use the 5G NR values as a concrete instance in this work and specify them for our evaluation.

\begin{table}[t]
\caption{MCS parameters for PDSCH under 3GPP NR~\cite{3gppTS38.214}\label{tab:mcs}}
\begin{center}
    {\renewcommand{\arraystretch}{1.1}%
    \small
        \begin{tabular}{c|c|c|c}
        \hline\hline
        $m$ & $Q_m$ & $R_m$ & Spectral efficiency (bits/RE) \\
        \hline
        0 & 2 & 602/1024 & 1.1758 \\
        1 & 4 & 658/1024 & 2.5703 \\
        2 & 6 & 873/1024 & 5.1152 \\
        3 & 8 & 948/1024 & 7.4063 \\
        \hline\hline
        \end{tabular}
    } 
\end{center}
\end{table}

\subsection{Fronthaul rate requirements}
\label{subsec:fronthaul}

The fronthaul transport consumes power that scales with the number of optical links a RU activates, so we compute the fronthaul rate here to determine that link count for the power model in \secref{sec:power_model}. The vDU transmits frequency-domain symbols to each RU via the eCPRI fronthaul link under Open RAN functional split 7-2x, as illustrated in~\figref{fig:ORAN}. Let $R_{\text{FH},\ell}$ denote the fronthaul data rate for RU $\ell$ and $N_{\text{bits}}^{\text{IQ}}$ the number of bits per in-phase and quadrature (IQ) sample after compression. Under joint coherent transmission, each active RU shares the same PRB allocation $n$. The overhead factor $\alpha_{\text{OH}}$ accounts for eCPRI protocol overhead, Ethernet framing, and cyclic redundancy check. We compute the downlink fronthaul rate for RU $\ell$ as
\begin{equation}
    R_{\text{FH},\ell} = \frac{ n \NRE^{\text{PRB}}  N_{\text{bits}}^{\text{IQ}} r}{T_{\text{slot}}} (1 + \alpha_{\text{OH}}).
\label{eq:fronthaul_rate}
\end{equation}
To reflect commercial Open RAN deployments, we assume optical links of capacity $C_{\text{link}} = 25$~Gbps using SFP28 transceivers. The vDU allocates an integer number of physical links to carry the fronthaul rate $R_{\text{FH},\ell}$, so RU $\ell$ uses
\begin{equation}
N_{\text{link},\ell} = \left\lceil \frac{R_{\text{FH},\ell}}{C_{\text{link}}} \right\rceil
\label{eq:n_link}
\end{equation}
optical links, giving a total capacity $C_{\text{FH},\ell} = N_{\text{link},\ell} C_{\text{link}}$. By construction $C_{\text{FH},\ell} \geq R_{\text{FH},\ell}$, so the fronthaul capacity constraint holds. Each link draws a fixed port power, so this link count sets the per-RU fronthaul power in \secref{sec:power_model}.

\begin{figure*}[t]
    \centering
    \includegraphics[width=0.90\textwidth]{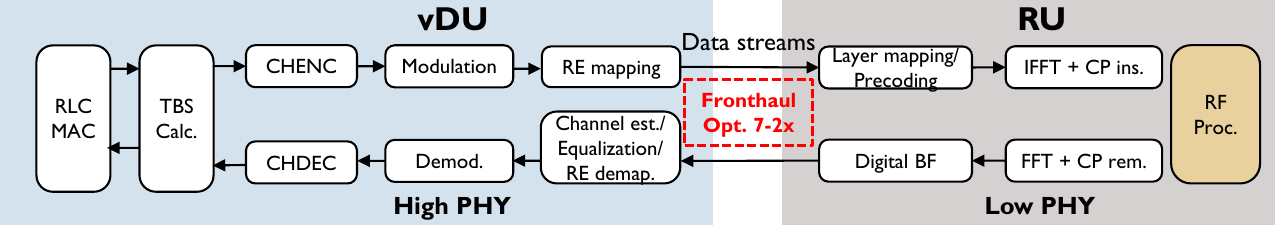}
    \caption{Open RAN functional split 7-2x. The vDU performs high PHY processing (channel encoding, modulation, precoding) and transmits frequency-domain symbols via eCPRI. The RU performs low PHY processing (IFFT, cyclic prefix insertion) before RF transmission. This split reduces fronthaul bandwidth by roughly an order of magnitude over CPRI.}
    \label{fig:ORAN}
\end{figure*}

\section{Power consumption model for distributed MIMO vRAN}
\label{sec:power_model}

We decompose the D-MIMO vRAN power into the three components introduced in \secref{sec:system_model}: vDU baseband processing, fronthaul transport, and per-RU power. This power model captures how the MCS, PRB allocation, rank, and transmit power drive each component, which the EE optimization in \secref{sec:EE_DMIMO} then minimizes.

We decompose the total power into vDU baseband processing on general-purpose servers, fronthaul optical transport, and RU RF chain and power amplifier (PA) operations as
\begin{equation}
    P_{\text{total}} = P_{\text{vDU}} + P_{\text{FH}} + P_{\text{RU}}.
\label{eq:power_total}
\end{equation}
The MCS index $m$, PRB allocation $n$, and rank $r$ take values in the feasible sets $\cM = \{0,1,2,3\}$, $\cN = \{1,\ldots,\NRB\}$, and $\cR = \{1,\ldots,\Ns\}$. 

We model the vDU power consumption as static baseline power and dynamic processing power. The baseline power $P_{\text{vDU},0}$ accounts for the host operating system, hypervisor, and idle hardware components. The dynamic power scales linearly with $n$ and $r$, since baseband processing complexity grows with FFT operations, resource element mapping, and channel coding over the allocated bandwidth. We denote $\Delta P_{\text{vDU}}$ as the power range between idle and fully loaded server operation, and $\alpha_{\text{CPU}} = 1 / (\NRB \cdot \Ns)$ as the normalized utilization coefficient per PRB per stream, so that $\alpha_{\text{CPU}} \cdot n \cdot r \in [0, 1]$ represents the fractional CPU load. The vDU power is
\begin{equation}
    P_{\text{vDU}}(n, r) = P_{\text{vDU},0} + \Delta P_{\text{vDU}} \cdot \alpha_{\text{CPU}} \cdot n \cdot r.
\label{eq:power_vdu}
\end{equation}

The fronthaul comprises an Ethernet switch and active optical links. The Ethernet switch connects the vDU to all RU fronthaul links, with power consumption split between a fixed chassis component and a per-port component that scales with the number of active links. We denote $P_{\text{switch},0}$ as the switch chassis static power, which accounts for the control plane, switching fabric, and cooling. Each active optical link consumes power $P_{\text{port}}$ for the Ethernet switch port and SFP28 transceiver pair. Each RU $\ell$ requires $N_{\text{link},\ell}(n, r)$ optical links from \eqref{eq:n_link}, which scales with $n$ and $r$ through the fronthaul rate $R_{\text{FH},\ell}$ in \eqref{eq:fronthaul_rate}. The total fronthaul power is
\begin{equation}
    P_{\text{FH}}(n, r) = P_{\text{switch},0} + \sum_{\ell=1}^{L} N_{\text{link},\ell}(n, r) \cdot P_{\text{port}}.
\label{eq:power_fh}
\end{equation}
The ceiling in $N_{\text{link},\ell}$ produces a staircase dependence on $R_{\text{FH},\ell}$, so the per-port contribution scales with each RU's link requirement independently while the chassis power $P_{\text{switch},0}$ remains constant.

Each RU deploys $\Nt$ transmit antennas, each with a dedicated RF chain. In MIMO mode all $\Nt$ chains are active, while in SIMO mode only one chain operates. Let $P_{\text{LO}}$ denote the local oscillator power shared across all chains. The static RU power per RU is
\begin{equation}
P_{\text{RU},0} = P_{\text{LO}} + \Nt \cdot P_{\text{RF}}.
\label{eq:p_ru0}
\end{equation}
We denote $P_{\ell}[k]$ as the radiated transmit power of RU $\ell$ at subcarrier $k$. We assume a uniform PA efficiency $\eta_{\text{PA}} \in (0,1]$ across all RUs, so a radiated power $P_{\text{tx},\ell}$ requires DC power $P_{\text{tx},\ell}/\eta_{\text{PA}}$. We also assume sufficient power-amplifier back-off for linear operation under the OFDM peak-to-average power ratio, and $\eta_{\text{PA}}$ captures the back-off efficiency. The total RU power across all $L$ RUs is
\begin{equation}
P_{\text{RU}} = \sum_{\ell=1}^{L} \left( P_{\text{RU},0} + \frac{1}{\eta_{\text{PA}}} P_{\text{tx},\ell}\right) = L \cdot P_{\text{RU},0} + \frac{1}{\eta_{\text{PA}}} \sum_{\ell=1}^{L} P_{\text{tx},\ell}.
\label{eq:power_ru}
\end{equation}

Each RU has a dedicated PA with its own power budget $P_{\max}$, so the per-subcarrier power $P_{\ell}[k]$ is an independent optimization variable for each RU $\ell$. We define the per-RU power vector $\bp_\ell = [P_\ell[0], \ldots, P_\ell[K-1]]^T \in \bbR_+^K$, so that $P_{\text{tx},\ell} = \mathbf{1}^T \bp_\ell$. The transmit power constraint applies to each RU individually rather than to the sum across RUs,
\begin{equation}
\mathbf{1}^T \bp_\ell \leq P_{\max}, \quad \forall \ell \in \{1, \ldots, L\}.
\label{eq:per_ru_constraint}
\end{equation}
The constraint in \eqref{eq:per_ru_constraint} bounds the total transmit power of each RU. Substituting \eqref{eq:power_vdu}, \eqref{eq:power_fh}, and \eqref{eq:power_ru} into \eqref{eq:power_total} with $P_{\text{tx},\ell} = \mathbf{1}^T \bp_\ell$, the total power decomposes as
\begin{equation}
    P_{\text{total}}\!\left(n, r, \{\bp_\ell\}\right) = P_{\text{static}}(n, r) + \frac{1}{\eta_{\text{PA}}} \sum_{\ell=1}^{L} \mathbf{1}^T \mathbf{p}_\ell.
\label{eq:power_decompose}
\end{equation}
The static power is
\begin{equation}
\begin{aligned}
    &P_{\text{static}}(n, r) = \underbrace{P_{\text{vDU},0} + \Delta P_{\text{vDU}} \cdot \alpha_{\text{CPU}} \cdot n \cdot r}_{P_{\text{vDU}}(n,r)} \\
    &+ \underbrace{P_{\text{switch},0} + \sum_{\ell=1}^{L} N_{\text{link},\ell}(n, r) \cdot P_{\text{port}}}_{P_{\text{FH}}(n, r)} + \underbrace{L \cdot \left( P_{\text{LO}} + \Nt \cdot P_{\text{RF}} \right)}_{L \cdot P_{\text{RU},0}}.
\end{aligned}
\label{eq:power_static}
\end{equation}
All three components of \eqref{eq:power_static} correspond to the vDU, fronthaul, and aggregated RU contributions, independent of the per-RU power vectors.

\section{Energy efficiency optimization for distributed MIMO vRAN}
\label{sec:EE_DMIMO}

We develop a cross-layer EE optimization framework for D-MIMO vRAN under a fixed traffic load. Our approach couples MAC-layer TBS calculation with PHY-layer power allocation through the power model of \secref{sec:power_model}. We first build the EESM-based link abstraction in \secref{subsec:eesm} and \secref{subsec:calibration}. We then formulate EE maximization as a convex problem and derive the closed-form power allocation in \secref{subsec:ee_problem} and \secref{subsec:optimal_solution}.

\subsection{Link abstraction via effective SNR mapping}
\label{subsec:eesm}

We develop an approach for selecting the communication parameters that maximize the rate in the system. This selection requires a link abstraction that predicts BLER for each candidate MCS and rank without per-realization link-level simulation. A machine learning model offers one such abstraction~\cite{DanielsHeathOnlineAdaptiveModulation2010}. Well-known effective SNR mapping functions offer another, with capacity, mutual-information, and exponential variants~\cite{HeathLozanoFoundationsMIMOCommunications2018}. We adopt EESM because its exponential form admits a convex reformulation and a closed-form power allocation that the other variants do not.

Effective SNR mapping compresses the per-subcarrier SNR vector $(\gamma[0], \ldots, \gamma[K-1])$ into a scalar equivalent SNR. This mapping applies a monotone function to each per-subcarrier SNR and inverts the average to recover the scalar equivalent SNR. The mapped SNR predicts BLER by matching the frequency selective channel to an AWGN channel at that SNR, avoiding per-realization link-level simulation. We denote $\gamma[k]$ as the single-stream SNR at subcarrier $k$. EESM uses an exponential mapping calibrated through a parameter $\beta > 0$~\cite{BlankenshipEtAlLinkErrorPrediction2004a}. The EESM effective SNR is
\begin{equation}
    \snreesm \approx -\beta \ln\!\left( \frac{1}{K+\Ncp} \sum_{k=0}^{K-1} \exp\!\left(-\frac{\gamma[k]}{\beta}\right) \right).
\label{eq:esm_general}
\end{equation}
The factor $1/(K+\Ncp)$ accounts for the cyclic prefix overhead in OFDM symbol duration. Capacity-effective SNR mapping (CESM) and mutual-information effective SNR mapping (MIESM), by contrast, use compression functions whose inverses lack closed-form expressions. The resulting SNR constraint is non-convex, precluding the closed-form KKT-based power allocation that the EESM constraint enables in \secref{subsec:ee_problem}.

We substitute the per-stream SNR in \eqref{eq:per_stream_snr} into \eqref{eq:esm_general} and aggregate over all $r$ active streams. With $a_{q,k} = \sigma_q^2[k]/(\beta r N_0 \Delta f)$, the EESM effective SNR as a function of $\bp \in \bbR_+^K$ is
\begin{equation}
\text{SNR}_{\text{EESM}}(\bp) = -\beta \ln\!\left( \frac{1}{r(K+\Ncp)} \sum_{q=1}^{r}\sum_{k=0}^{K-1} e^{-a_{q,k} P[k]} \right).
\label{eq:eesm}
\end{equation}
The aggregation over all $r$ streams follows the 3GPP NR layer-to-codeword mapping rule. 3GPP NR maps a single transport block onto up to four spatial layers via a single LDPC codeword~\cite{3gppTS38.211}, which covers all configurations in this work. The modulated symbols of one codeword are layer-mapped across all active streams, and decoding succeeds or fails jointly for the entire codeword. The joint $\text{SNR}_{\text{EESM}}$ in \eqref{eq:eesm} therefore matches the standard-compliant decoding behavior under single-codeword mapping.

From \eqref{eq:ninfo}, we redefine the TBS for MCS index $m$, $n$ allocated PRBs, and $r$ spatial streams as
\begin{equation}
    \NTBS (m, n, r) = N_{\text{RE}}^{\text{PRB}} R_m Q_m n r,
\label{eq:tbs}
\end{equation}
which represents the number of information bits per slot. Let $p_e(\cdot)$ denote the AWGN BLER curve for MCS index $m$ obtained from offline link-level simulation. The effective throughput is
\begin{equation}
R_{\text{eff}}(m, n, r) = \frac{\left(1 - p_e\!\left(\text{SNR}_{\text{EESM}}\right)\right) \cdot \NTBS(m, n, r)}{T_{\text{slot}}}.
\label{eq:throughput_eff}
\end{equation} We adopt a step-function approximation of the BLER curve for tractable optimization. Let $p_e^{\text{target}}$ denote the target BLER that defines successful decoding, and let $\gamma_{\text{th}}^{(m)}$ denote the SNR at which $p_e(\gamma_{\text{th}}^{(m)}) = p_e^{\text{target}}$. The step-function approximation~is
\begin{equation}
    p_e(\text{SNR}) \approx \begin{cases}
    1, & \text{if } \text{SNR} < \gamma_{\text{th}}^{(m)}, \\
    0, & \text{if } \text{SNR} \geq \gamma_{\text{th}}^{(m)}.
\end{cases}
\label{eq:bler_step}
\end{equation}
The effective throughput equals $\NTBS(m, n, r) / T_{\text{slot}}$ whenever $\text{SNR}_{\text{EESM}}(\bp) \geq \gamma_{\text{th}}^{(m)}$ and zero otherwise. We transform the BLER constraint into a deterministic feasibility condition for convex optimization in \secref{subsec:ee_problem}.

\subsection{Offline EESM calibration}
\label{subsec:calibration}

Offline EESM calibration is a standard procedure in system-level evaluation, where a per-MCS $\beta$ is fit once against link-level BLER curves and then reused across deployments. Industry simulation libraries also adopt the same method, with NVIDIA Sionna drawing its EESM coefficients directly from~\cite{LagenEtAlNewRadioPhysical2020, sionna}. Rather than reusing these published coefficients, we fit $\beta$ in \eqref{eq:eesm} for each MCS index $m$ on our own 3GPP NR PDSCH chain via \algref{alg:beta_calibration}, executed in the link-level simulation prior to deployment under single-stream transmission. The calibrated $\beta_m$ ensures that the EESM-predicted BLER matches the actual BLER across all TDL channel realizations and operating SNR points. The single-stream calibration directly supports the closed-form power allocation for $r=1$ in Theorem~\ref{thm:optimal_power}. Multi-stream extension may benefit from per-rank $\beta_m^{(r)}$ calibration through repeated application of \algref{alg:beta_calibration} with the target $r$ in \eqref{eq:eesm}.

The calibrated parameter $\beta_m$ is obtained by averaging over all realizations and SNR points,
\begin{equation}\beta_m = \frac{1}{N_{\text{SNR}} N_{\text{cal}}} \sum_{i=1}^{N_{\text{SNR}}} \sum_{j=1}^{N_{\text{cal}}} \beta^{(i,j)}.
\label{eq:beta_opt}
\end{equation}
The resulting scalar $\beta_m$ is stored at the vDU for use during runtime MCS selection. At runtime, the vDU retrieves $\beta_m$ by interpolating $\{\beta_m^{(i)}\}$ at the SNR estimated from UE channel quality indicator (CQI) feedback, and treats it as fixed throughout the optimization to preserve the convex structure.

\begin{algorithm}[t]
\caption{Offline EESM calibration for MCS index $m$}
\label{alg:beta_calibration}
\begin{algorithmic}[1]
\STATE \textbf{Input:} MCS index $m$, SNR grid $\{\gamma_{\text{SNR}}^{(i)}\}_{i=1}^{N_{\text{SNR}}}$, number of channel realizations per SNR $N_{\text{cal}}$
\STATE \textbf{Output:} Calibrated parameter $\beta_m$
\STATE Generate AWGN BLER reference curve $p_e(\gamma)$ for MCS $m$ by running link-level simulation
\FOR{$i = 1$ to $N_{\text{SNR}}$}
    \STATE Set operating SNR $\gamma_{\text{SNR}} \gets \gamma_{\text{SNR}}^{(i)}$
    \STATE Initialize $\beta_{\text{snr}} \gets \{\}$
    \FOR{$j = 1$ to $N_{\text{cal}}$}
        \STATE Generate TDL channel realization $\bH^{(j)}$
        \STATE Compute per-subcarrier SNRs $\{\gamma[k]^{(j)}\}_{k=0}^{K-1}$ at $\gamma_{\text{SNR}}$
                \STATE Run full link-level simulation to obtain $\text{BLER}^{(j)}$
        \IF{$\text{BLER}^{(j)} \in [10^{-3},\, 0.5]$}
            \STATE Find $\gamma_{\text{target}}^{(j)}$ such that $p_e(\gamma_{\text{target}}^{(j)}) = \text{BLER}^{(j)}$
            \STATE Numerically solve for $\beta^{(j)}$ such that $\gamma_{\text{EESM}}(\beta^{(j)}) = \gamma_{\text{target}}^{(j)}$ using \eqref{eq:eesm}
            \STATE $\beta_{\text{snr}} \gets \beta_{\text{snr}} \cup \{\beta^{(j)}\}$
        \ENDIF
    \ENDFOR
    \STATE $\beta_m^{(i)} \gets \frac{1}{|\beta_{\text{snr}}|} \sum_{\beta \in \beta_{\text{snr}}} \beta$
\ENDFOR
\STATE $\beta_m \gets \frac{1}{N_{\text{SNR}}} \sum_{i=1}^{N_{\text{SNR}}} \beta_m^{(i)}$
\RETURN $\beta_m$
\end{algorithmic}
\end{algorithm}

\subsection{Energy efficiency optimization problem}
\label{subsec:ee_problem}

In this section, we formulate energy efficiency optimization using the EESM prior work of \secref{subsec:eesm} and \secref{subsec:calibration}, which focuses only on rate. We therefore define a new metric that couples the predicted throughput to the power model of \secref{sec:power_model}. We pose this problem for a fixed configuration $(m, n, r) \in \cM \times \cN \times \cR$. 

We define the system-level energy efficiency metric as the ratio of effective throughput to total power consumption as
\begin{equation}
\eta_{\text{EE}}(n, r, \bp) = \frac{R_{\text{eff}}(m, n, r)}{P_{\text{total}}(n, r, \bp)}.
\label{eq:ee_metric}
\end{equation}
Here, $L_{\text{act}}$ denotes the number of active RUs sharing the common power vector $\bp$. We substitute the effective throughput from \eqref{eq:throughput_eff} and the total power decomposition from \eqref{eq:power_decompose} into \eqref{eq:ee_metric} to obtain 
\begin{equation}
        \eta_{\text{EE}}(n, r, \bp) = \frac{\NTBS(m, n, r)}{T_{\text{slot}}} \cdot \frac{1 - p_e(\text{SNR}_{\text{EESM}}(\bp))}{P_{\text{static}}(n, r) + \frac{L_{\text{act}}}{\eta_{\text{PA}}} \mathbf{1}^T \bp}.
\label{eq:ee_expanded}
\end{equation}
Under the step-function approximation in \eqref{eq:bler_step}, the effective throughput is $\NTBS(m, n, r) / T_{\text{slot}}$ or zero. For a fixed configuration $(m, n, r)$, the numerator $\NTBS(m,n,r)/T_{\text{slot}}$ is a constant, $P_{\text{static}}(n, r)$ is a fixed offset independent of $\bp$, and $L_{\text{act}}$ is constant within the configuration. Maximizing $\eta_{\text{EE}}$ therefore reduces to minimizing the per-RU transmit power $\mathbf{1}^T\bp$. Following the common per-subcarrier power assignment in \secref{subsec:eesm}, we drop the RU index $\ell$ since the same solution applies to all active RUs. We formulate the optimization problem as
\begin{equation} \label{eq:power_min}
    \min_{\bp} \quad \mathbf{1}^T \bp 
\end{equation}
\begin{subequations} \label{eq:consts}
    \begin{align}
        \text{subject to} \quad 
        & \text{SNR}_{\text{EESM}}(\bp) \geq \gamma_{\text{th}}^{(m)}, \label{eq:const_eesm}\\
        & \mathbf{1}^T \bp \leq P_{\max}, \label{eq:const_sumpower} \\
        & \bp \succeq \mathbf{0}, \label{eq:const_power}
    \end{align}
\end{subequations}
where $P_{\max}$ is the maximum per-RU transmit power budget and $\bp \succeq \mathbf{0}$ denotes element-wise non-negativity.

We transform the EESM constraint into a convex form. Substituting the multi-stream EESM expression \eqref{eq:eesm} into the constraint $\text{SNR}_{\text{EESM}}(\bp) \geq \gamma_{\text{th}}^{(m)}$ and rearranging gives
\begin{equation}
-\beta \ln\!\left( \frac{1}{r(K+\Ncp)} \sum_{q=1}^{r} \sum_{k=0}^{K-1} \exp(-a_{q,k} P[k]) \right) \geq \gamma_{\text{th}}^{(m)}.
\label{eq:eesm_constraint_raw}
\end{equation}
Multiplying both sides by $-1/\beta$ and applying the exponential function yields the equivalent constraint
\begin{equation}
\frac{1}{r(K+\Ncp)} \sum_{q=1}^{r} \sum_{k=0}^{K-1} \exp(-a_{q,k} P[k]) \leq \exp\!\left(-\frac{\gamma_{\text{th}}^{(m)}}{\beta}\right).
\label{eq:eesm_convex_form}
\end{equation}
Equation \eqref{eq:power_min} is convex since the objective is linear and the multi-stream EESM constraint \eqref{eq:eesm_convex_form} defines a convex set as a sublevel set of a nonnegative sum of exponential functions, each of which is convex in $P[k]$. The global optimum is therefore characterized by KKT conditions.

\subsection{Optimal power allocation}
\label{subsec:optimal_solution}

Maximizing EE reduces to meeting the EESM threshold with the least transmit power. Subcarriers differ in channel gain, so a uniform power level wastes power. We therefore allocate power per subcarrier by its channel gain, which meets the threshold at a lower total power and thus a higher EE. The vDU applies this allocation at each slot as a practical runtime rule, not only a performance bound. The optimal power allocation follows from the KKT conditions applied to \eqref{eq:power_min}. Recalling $a_{q,k}$ from \eqref{eq:eesm}, we further define $c = \exp(-\gamma_{\text{th}}^{(m)} / \beta)$ for compactness. The multipliers and the calibration parameter depend on the MCS index $m$.  We omit the $m$ subscript during the derivation for notational simplicity, restoring it in the final closed form. Let $\lambda \geq 0$ be the multipliers for the EESM constraint, $\nu \geq 0$ for the total power constraint, and $\boldsymbol{\xi} \geq \mathbf{0}$ for the non-negativity constraints. The Lagrangian is
\begin{equation}
\begin{aligned}
    \mathcal{L} &= \mathbf{1}^T \bp + \lambda \left( \frac{1}{r(K+\Ncp)} \sum_{q=1}^{r} \sum_{k=0}^{K-1} e^{-a_{q,k} P[k]} - c \right) \\
    &\quad + \nu \left( \mathbf{1}^T \bp - P_{\max} \right) - \boldsymbol{\xi}^T \bp.
\end{aligned}
\end{equation}
The KKT stationarity condition is
\begin{equation}
\frac{\partial \mathcal{L}}{\partial P[k]} = 1 + \nu - \frac{\lambda}{r(K+\Ncp)} \sum_{q=1}^{r} a_{q,k}\, e^{-a_{q,k} P[k]} - \xi_k = 0.
\label{eq:kkt_stationarity}
\end{equation}
For $r > 1$, the sum of exponentials in \eqref{eq:kkt_stationarity} yields a transcendental equation in $P^*[k]$ that must be solved numerically. For the single-stream case $r = 1$, however, the sum collapses to a single exponential and admits a closed-form solution, stated in Theorem~\ref{thm:optimal_power}.

\begin{theorem}
\label{thm:optimal_power}
Let $\cA = \{k : P^*[k] > 0\}$ denote the active set, and $\cI = \{0, \ldots, K{-}1\} \setminus \cA$ the inactive subcarrier set. For the single-stream case ($r = 1$) with feasible $(m,n,r)$, the optimal power allocation for $k \in \cA$ is
\begin{equation}
P^*[k] = \frac{\beta N_0 \Delta f}{\sigma_1^2[k]} \ln\!\left( \frac{\lambda \sigma_1^2[k]}{(K+\Ncp) \beta N_0 \Delta f} \right),
\label{eq:optimal_power}
\end{equation}
where $\lambda > 0$ satisfies the EESM constraint with equality over $\cA$. The closed form is
\begin{equation}
\lambda_m = \frac{\beta_m N_0 \Delta f}{\exp(-\gamma_{\text{th}}^{(m)} / \beta_m) - |\cI|/(K+\Ncp)} \sum_{k \in \cA} \frac{1}{\sigma_1^2[k]}.
\label{eq:lambda_closed}
\end{equation}
Subcarriers with $\sigma_1^2[k]$ below the threshold determined by $\lambda$ are inactive. See Appendix~\ref{app:proof_theorem1} for the active set determination.
\end{theorem}
\begin{IEEEproof}
See Appendix~\ref{app:proof_theorem1}.
\end{IEEEproof}

\begin{remark}
The allocation in \eqref{eq:optimal_power} is waterfilling-like: a larger $\sigma_1^2[k]$ needs less transmit power, and $\sigma_1^2[k]$ carries the coherent combining gain across the $L$ distributed RUs. Classical waterfilling sets its water level from a total power budget. Here the water level $\lambda$ in \eqref{eq:lambda_closed} follows from the EESM threshold $\gamma_{\text{th}}^{(m)}$ and the calibration parameter $\beta$, which fixes a unique optimal total power analyzed in Section~\ref{sec:results}.
\end{remark}


\section{3D spatio-temporal traffic-aware optimization for D-MIMO}
\label{sec:traffic-aware-d-mimo}

In this section, we develop a control strategy that adapts the RU operation mode to save energy under spatio-temporal traffic dynamics. We extend the fixed-load framework of~\secref{sec:EE_DMIMO} by switching each RU among MIMO, SIMO, and sleep mode as its traffic demand varies. The framework reuses its closed-form power allocation at each scheduling interval. We define the three RU operation modes, derive the demand thresholds that separate them, and present the control algorithm that selects the mode, MCS, and power allocation at each scheduling interval. A mode control strategy needs traffic that varies in space and time. We therefore build a traffic model that captures this variation from realistic user mobility.

\begin{figure}[t]
    \centering
    \subfloat[]{\includegraphics[width=0.90\columnwidth]{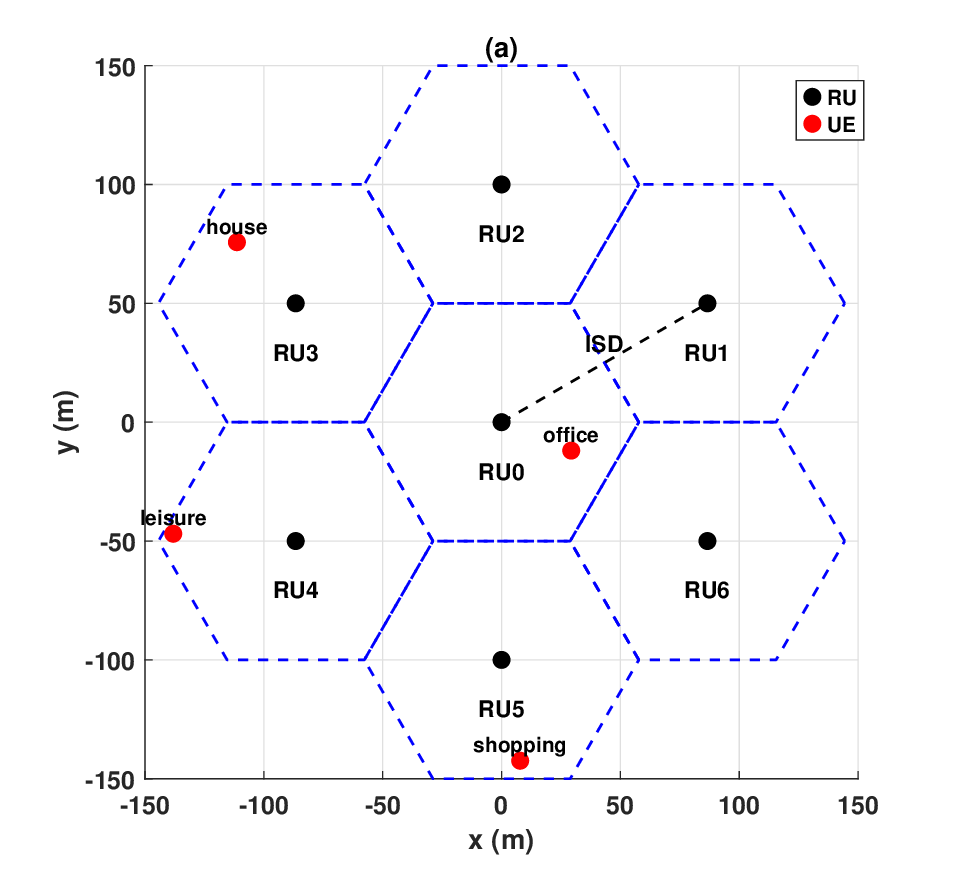}}%
    \label{fig:fig4a_hex_trajectory}
    \\
    \subfloat[]{\includegraphics[width=0.90\columnwidth]{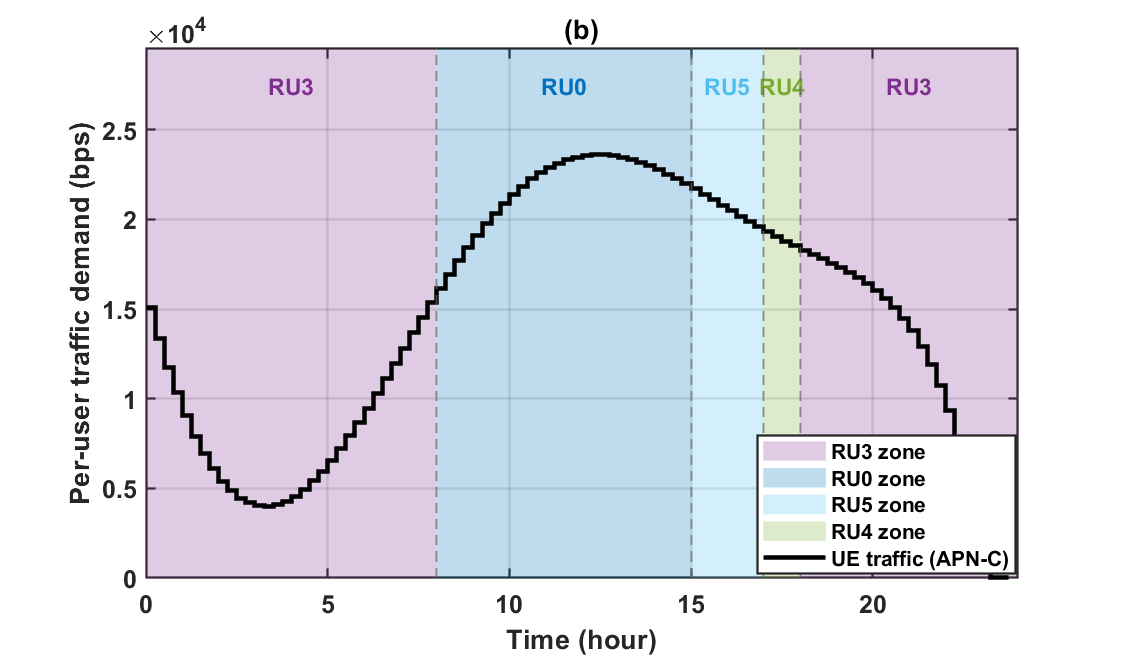}}%
    \label{fig:fig4b_ue_traffic}
    \caption{(a) UMi hexagonal cell layout with $L = 7$ RUs at $d_{\text{ISD}} = 100$~m inter-site distance and a representative daily UE trajectory from MATSim~\cite{AxhausenETHZurichMultiagentTRANsportSimulation2016}, with red markers at activity locations. (b) Per-user traffic demand $\Lambda(t)$ under the APN-C profile over a 24 hour period. Background shading indicates the RU coverage zone occupied by the UE. Diurnal variation and spatial transitions motivate the adaptive RU mode control strategy.}
    \label{fig:hexamap_3d}
\end{figure}

Instantaneous traffic load at each RU determines whether it operates in full MIMO, reduced-rank SIMO, or sleep mode. We deploy RUs over an urban micro (UMi) hexagonal cell layout with inter-site distance $d_{\text{ISD}}$, as shown in Fig.~\ref{fig:hexamap_3d}(a). The RUs operate cell-free under a common physical cell identity, so the UE perceives the deployment as a single cell. At each scheduling instant $t$, the UE associates with its $L_{\text{act}}$ nearest RUs for coherent joint transmission, while the vDU coordinates resource allocation across all active RUs.

The wireless literature offers no standardized model for joint spatial and temporal traffic variation. We construct one by modifying and combining a mobility simulator with a published demand profile. MATSim provides realistic UE trajectories from daily activity plans at the individual agent level~\cite{AxhausenETHZurichMultiagentTRANsportSimulation2016}, capturing spatial patterns that standard stochastic mobility models miss. The Access Point Name (APN) C profile from~\cite{SoosEtAlUserGroupBehavioural2020} provides empirical temporal demand statistics for smartphone subscribers. The two sources lack natural alignment, so we bridge them with two assumptions. The APN-C temporal shape applies to each user independently of location. All spatial demand structure follows from user mobility. We fit a degree-6 polynomial $f(t)$ to the APN-C profile in 15-minute intervals from~\cite[Table~III]{SoosEtAlUserGroupBehavioural2020}. The per-user demand $\Lambda(t)$ scales the polynomial to the 5G NR user experienced peak data rate $R_{\text{peak}}$ as
\begin{equation}
    \Lambda(t) = R_{\text{peak}} \cdot \frac{f(t)}{\max_{t'} f(t')}.
\end{equation}
Figure~\ref{fig:hexamap_3d}(b) shows $\Lambda(t)$ over one day with background shading indicating the RU coverage zone occupied by the UE at each time interval. The APN-C profile exhibits strong diurnal variation, with off-peak demand near midnight falling to roughly 15\% of the daytime peak. As the UE transitions across RU zones, the traffic load shifts among RUs, creating spatially non-uniform and time-varying demand. This structure motivates the adaptive RU mode control strategy in Algorithm~\ref{alg:joint}.

We define three operation modes for each RU $\ell$. MIMO mode activates all $\Nt$ transmit antennas and serves the UE at full rank $r = \Ns$ via SVD precoding. SIMO mode transmits a single stream ($r = 1$) with one RF chain active. Sleep mode deactivates the RF chain and PA, consuming only the idle power $P_{\text{idle}} \ll P_{\text{RU},0}$. Let $\Lambda_\ell(t)$ denote the traffic demand served by RU $\ell$ at time $t$, determined by the UE's association with its $L_{\text{act}}$ nearest RUs. The mode-specific static RU power from \eqref{eq:p_ru0} is $P_{\text{RU},0}^{(\text{MIMO})} = P_{\text{LO}} + \Nt \cdot P_{\text{RF}}$ in MIMO mode and $P_{\text{RU},0}^{(\text{SIMO})} = P_{\text{LO}} + P_{\text{RF}}$ in SIMO mode. We approximate the staircase fronthaul power in \eqref{eq:power_fh} with a linear model, denoting by $\eta_{\text{FH}}$~[W/(stream$\cdot$PRB)] the marginal fronthaul power per additional IQ stream. The per-RU RF, PA, and fronthaul power under each mode is
\begin{equation}
\begin{aligned}
&\begin{cases} 
P_{\text{RU},\ell}^{(\text{mode})} = P_{\text{RU},0}^{(\text{MIMO})} + \tfrac{1}{\eta_{\text{PA}}} P_{\text{tx},\ell}^{(\text{MIMO})}, \\ P_{\text{FH},\ell}^{(\text{mode})} = P_{\text{port}} + \eta_{\text{FH}} \cdot \Ns \cdot n, \end{cases} & \Lambda_\ell(t) \geq \Lambda^*, \\
&\begin{cases} 
P_{\text{RU},\ell}^{(\text{mode})} = P_{\text{RU},0}^{(\text{SIMO})} + \tfrac{1}{\eta_{\text{PA}}} P_{\text{tx},\ell}^{(\text{SIMO})}, \\ P_{\text{FH},\ell}^{(\text{mode})} = P_{\text{port}} + \eta_{\text{FH}} \cdot n, \end{cases} & \Lambda^{\text{sleep}} \leq \Lambda_\ell(t) < \Lambda^*, \\
&\begin{cases} 
P_{\text{RU},\ell}^{(\text{mode})} = P_{\text{idle}}, \\ P_{\text{FH},\ell}^{(\text{mode})} = 0, \end{cases} & \Lambda_\ell(t) < \Lambda^{\text{sleep}}.
\end{aligned}
\label{eq:mode_power_combined}
\end{equation}
When RU $\ell$ enters sleep mode, the vDU deactivates the corresponding fronthaul link, terminating both $P_{\text{port}}$ and the IQ transport overhead.

The key trade-off between MIMO and SIMO is that MIMO achieves a higher $\NTBS(m, n, \Ns)$ through spatial multiplexing, but must meet $\gamma_{\text{th}}^{(m)}$ across all $\Ns$ active streams, driving up the minimum transmit power. The minimum transmit power under each mode follows from Theorem~\ref{thm:optimal_power} with the corresponding $\Ns$, denoted $P_{\text{tx},\ell}^{(\text{MIMO})}(\Lambda)$ and $P_{\text{tx},\ell}^{(\text{SIMO})}(\Lambda)$. The exact system-level crossover condition equates total vRAN power across the two modes as

\begin{equation}
\begin{aligned}
& L_{\text{act}}\!\left[ P_{\text{RU},0}^{(\text{MIMO})} + P_{\text{tx},\ell}^{(\text{MIMO})}(\Lambda^*)/\eta_{\text{PA}} + \eta_{\text{FH}} \Ns n \right] + P_{\text{vDU}}(n, \Ns) \\
& = L_{\text{act}}\!\left[ P_{\text{RU},0}^{(\text{SIMO})} + P_{\text{tx},\ell}^{(\text{SIMO})}(\Lambda^*)/\eta_{\text{PA}} + \eta_{\text{FH}} n \right] + P_{\text{vDU}}(n, 1).
\end{aligned}
\label{eq:crossover_full}
\end{equation}
Dropping the vDU and fronthaul terms from \eqref{eq:crossover_full} gives the simplified RU-only crossover condition
\begin{equation}
P_{\text{RU},0}^{(\text{MIMO})} + \frac{P_{\text{tx},\ell}^{(\text{MIMO})}(\Lambda^*)}{\eta_{\text{PA}}} = P_{\text{RU},0}^{(\text{SIMO})} + \frac{P_{\text{tx},\ell}^{(\text{SIMO})}(\Lambda^*)}{\eta_{\text{PA}}}.
\label{eq:crossover}
\end{equation}
The dropped terms are larger in MIMO mode, so \eqref{eq:crossover} underestimates $\Lambda^*$. The vDU evaluates the exact form \eqref{eq:crossover_full} for $\Lambda^*$ using the EESM-based transmit power from Theorem~\ref{thm:optimal_power}, while \eqref{eq:crossover} serves as the RU-only baseline in Section~\ref{sec:results}. Analogously, the SIMO to sleep crossover threshold $\Lambda^{\text{sleep}}$ equates SIMO total per-RU power with sleep mode power as
\begin{equation}
P_{\text{RU},0}^{(\text{SIMO})} + \frac{P_{\text{tx},\ell}^{(\text{SIMO})}(\Lambda^{\text{sleep}})}{\eta_{\text{PA}}} = P_{\text{idle}}.
\label{eq:sleep_crossover}
\end{equation}
Since $P_{\text{RU},0}^{(\text{SIMO})} > P_{\text{idle}}$, sleep mode is strictly more power-efficient than SIMO whenever $\Lambda_\ell(t) < \Lambda^{\text{sleep}}$.

\begin{algorithm}[t]
\caption{Per-slot MCS selection and power allocation}
\label{alg:joint}
\begin{algorithmic}[1]
\STATE \textbf{Input:} Mode assignment $\{\text{mode}_\ell\}$, channel singular values $\{\sigma_r^2[k]\}_{r \in \cR}$, calibrated parameters $\{\beta_m\}_{m \in \cM}$, SNR thresholds $\{\gamma_{\text{th}}^{(m)}\}_{m \in \cM}$, noise power per subcarrier $N_0 \Delta f$, power budget $P_{\max}$, PRB allocation $n$, traffic demand $\Lambda(t)$
\STATE \textbf{Output:} Optimal configuration $(m^*, r^*)$ and power allocation $\{P^*[k]\}$
\STATE Set $\cR(t) \leftarrow \{1\}$ if $\text{mode}_\ell = \text{SIMO}$, else $\cR(t) \leftarrow \cR$
\STATE Compute $\cM(t)$ from \eqref{eq:feasible_mcs} using $r_{\max}^{(\text{mode})}$ determined by $\text{mode}_\ell$
\IF{$\cM(t) = \emptyset$}
    \STATE Declare outage; terminate
\ENDIF
\FOR{each $(m, r) \in \cM(t) \times \cR(t)$}
    \STATE Compute $P_{\text{static}}(n, r)$ from \eqref{eq:power_static}
    \IF{$r = 1$}
        \STATE Compute $\lambda_m$ from \eqref{eq:lambda_closed}
        \STATE Compute $P^{(m,r)}[k]$ from \eqref{eq:optimal_power} for $k = 0, \ldots, K-1$
    \ELSE
        \STATE Numerically solve \eqref{eq:kkt_stationarity} together with the EESM equality constraint to obtain $P^{(m,r)}[k]$ for $k = 0, \ldots, K-1$
    \ENDIF
    \IF{$\sum_k P^{(m,r)}[k] > P_{\max}$}
        \STATE Set $\eta_{\text{EE}}^{(m,r)} = 0$
    \ELSE
        \STATE Evaluate $\eta_{\text{EE}}^{(m,r)}$ from \eqref{eq:ee_expanded}
    \ENDIF
\ENDFOR
\STATE $(m^*, r^*) = \argmax_{(m,r) \in \cM(t) \times \cR(t)} \; \eta_{\text{EE}}^{(m,r)}$
\STATE $P^*[k] = P^{(m^*, r^*)}[k]$ for $k = 0, \ldots, K-1$
\STATE \RETURN $(m^*, r^*)$, $\{P^*[k]\}_{k=0}^{K-1}$
\end{algorithmic}
\end{algorithm}

The vDU allocates resources in two stages at each scheduling interval. In the outer stage, the vDU evaluates the traffic demand $\Lambda_\ell(t)$ at each RU and assigns an operation mode per \eqref{eq:mode_power_combined}, using the vRAN-aware threshold from \eqref{eq:crossover_full}. In the inner stage, the vDU sets the feasible stream set to $\cR(t) = \{1\}$ for SIMO mode and $\cR(t) = \cR$ for MIMO mode, with maximum rank $r_{\max}^{(\text{mode})} = \max \cR(t)$. The vDU then computes the demand-feasible MCS set
\begin{equation}
\cM(t) = \left\{ m \in \cM : \frac{\NRE^{\text{PRB}}\, R_m\, Q_m\, n\, r_{\max}^{(\text{mode})}}{T_{\text{slot}}} \geq \Lambda(t) \right\},
\label{eq:feasible_mcs}
\end{equation}
which retains only the MCS indices whose peak throughput meets the instantaneous traffic demand $\Lambda(t)$. If $\cM(t) = \emptyset$, the vDU declares an outage for that slot. For each feasible $(m, r)$, the vDU obtains the per-subcarrier power from Theorem~\ref{thm:optimal_power} for SIMO and by solving \eqref{eq:kkt_stationarity} for MIMO. The vDU evaluates the resulting EE from \eqref{eq:ee_expanded} and selects the configuration with the highest EE. Algorithm~\ref{alg:joint} summarizes the full procedure.

\begin{table}[t]
\caption{D-MIMO vRAN simulation parameters\label{tab:sim_params}}
\begin{center}
    {\renewcommand{\arraystretch}{1.1}%
    \small
        \begin{tabular}{c|c|c|c}
        \hline\hline
        Notation & Value & Notation & Value \\
        \hline
        $f_c$ & 3.5 GHz & $P_{\max}$ & 33 dBm \\
        $\Delta f$ & 30 kHz ($\mu$=1) & $R_{\text{peak}}$ & 1 Gbps \\
        $K$ & 3276 & $\eta_{\text{PA}}$ & 0.38 \\
        $\NRB$ & 273 & $P_{\text{RU},0}^{(\text{MIMO})}$ & 2.0 W \\
        $N_{\text{FFT}}$ & 4096 & $P_{\text{RU},0}^{(\text{SIMO})}$ & 1.0 W \\
        $N_{\text{symb}}$ & 14 & $P_{\text{idle}}$ & 0.1 W \\
        $\Nt$ & 2 & $P_{\text{vDU},0}$ & 15 W \\
        $\Nr$ & 2 & $\Delta P_{\text{vDU}}$ & 15 W \\
        $L$ & 7 & $\alpha_{\text{CPU}}$ & $1/(\NRB \cdot \Ns)$ \\
        $L_{\text{act}}$ & 3 & $P_{\text{switch},0}$ & 8 W \\
        $d_{\text{ISD}}$ & 100 m & $P_{\text{port}}$ & 2 W \\
        $N_{\text{cal}}$ & 100 & $C_{\text{link}}$ & 25 Gbps \\
        $p_e^{\text{target}}$ & $10^{-2}$ & $N_0$ & $-174$ dBm/Hz \\
        \hline\hline
        \end{tabular}
    }
\end{center}
\end{table}

\section{Simulation results}
\label{sec:results}

\begin{figure}[t]
    \centering
    \subfloat[]{\includegraphics[width=0.95\columnwidth]{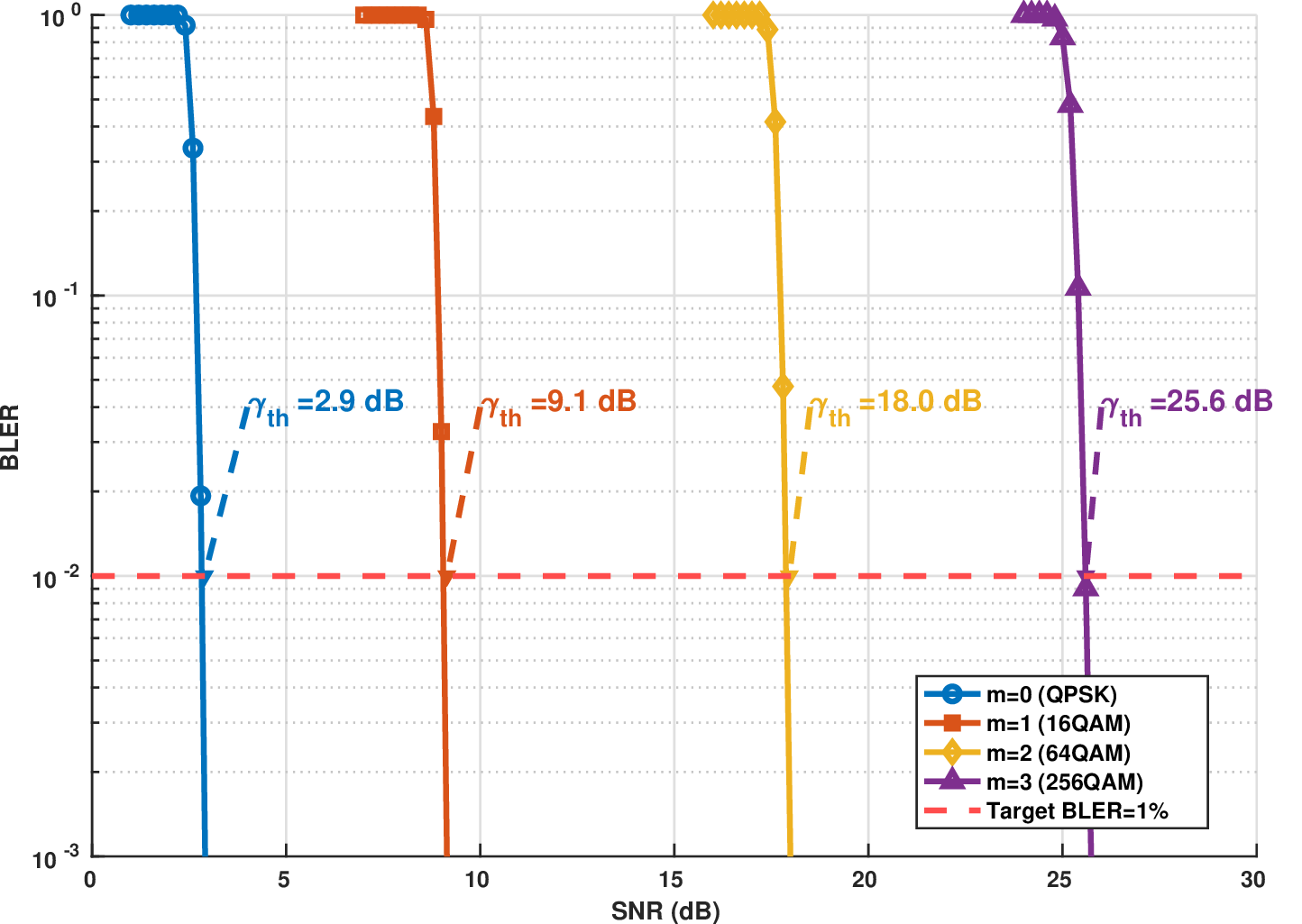}}%
    \label{fig:fig5a_awgn_bler}
    \\
    \subfloat[]{\includegraphics[width=0.95\columnwidth]{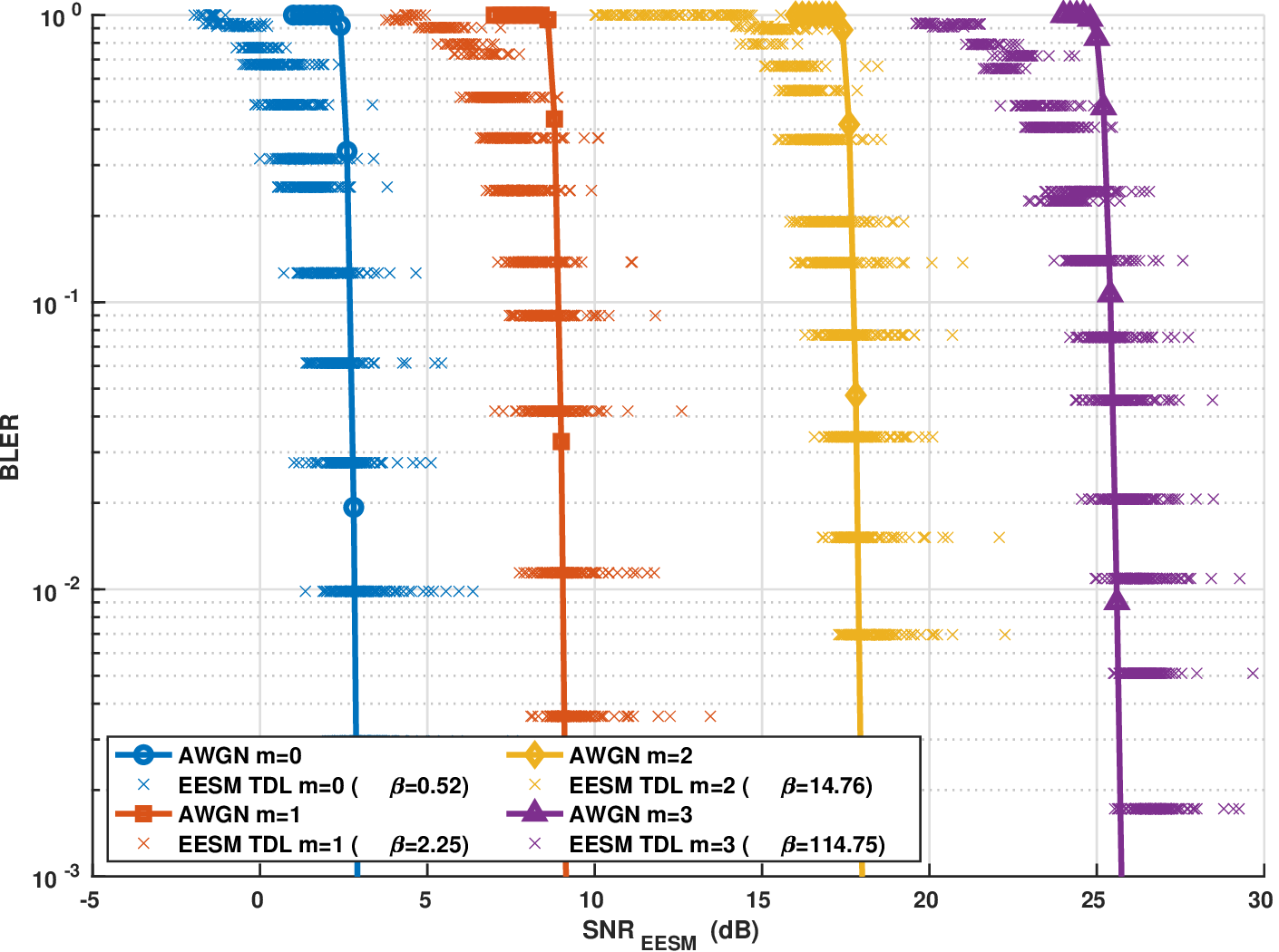}}%
    \label{fig:fig5b_beta_calib}
    \caption{(a) AWGN BLER curves for MCS 0, 1, 2, and 3 under single-input single-output (SISO) transmission with $N_{\text{RB}} = 273$. The SNR threshold $\gamma_{\text{th}}^{(m)}$ at $p_e^{\text{target}} = 10^{-2}$ ranges from 2.9~dB for QPSK to 25.6~dB for 256-QAM, serving as the AWGN reference for Algorithm~\ref{alg:beta_calibration}. (b) EESM calibration validation. The scatter points confirm that a single $\beta_m$ accurately maps the TDL channel to an equivalent AWGN point.}
    \label{fig:fig5_link_abstraction}
\end{figure}

We evaluate the proposed framework with three goals. First, we verify that a single offline-calibrated $\beta_m$ predicts BLER without per-realization link-level simulation. Second, we show that the vRAN-aware power model widens the traffic range where SIMO saves power beyond the RF-only estimate. Third, we show that traffic-aware mode control raises average EE. We developed a C++ based link-level simulator interfaced with MATLAB that implements the 3GPP NR PDSCH chain~\cite{3gppTS38.211, 3gppTS38.212, 3gppTS38.214}. The simulator includes LDPC coding, OFDM modulation, and LMMSE detection with perfect receiver-side channel knowledge, using $10^5$ realizations per SNR point. We draw the small-scale fading for each RU from an independent TDL-C realization, since the inter-RU spacing exceeds the UMi decorrelation distance~\cite{3gppTR38.901}. We calibrate the BLER curves against the 3GPP technical specifications for NR physical layer requirements. \tabref{tab:sim_params} summarizes the simulation parameters, following~\cite{AuerEtAlHowMuchEnergy2011} and our prior work~\cite{ParkEtAlSwitchDFTAdaptiveWaveform2026}. We set $n = \NRB$ since the single-user setting allocates full bandwidth to the active UE.

\subsection{EESM calibration accuracy}
\label{subsec:calibration_results}

We validate the EESM calibration by comparing EESM-predicted BLER against full link-level simulation results over TDL-C channels with a 100~ns delay spread. Figure~\ref{fig:fig5_link_abstraction}(a) shows the AWGN BLER reference curves for MCS 0, 1, 2, and 3, generated by sweeping the SNR over an AWGN channel. The SNR threshold $\gamma_{\text{th}}^{(m)}$ at target BLER $10^{-2}$ ranges from 2.9~dB for MCS~0 QPSK to 25.6~dB for MCS~3 256-QAM. Figure~\ref{fig:fig5_link_abstraction}(b) shows the calibration validation. For each MCS and SNR point, we generate $N_{\text{cal}} = 100$ independent TDL realizations and compute $\text{SNR}_{\text{EESM}}$ using $\beta_m$. The scatter points align with the AWGN reference curves across all MCS indices. A single $\beta_m$ therefore maps the frequency-selective TDL channel to an equivalent AWGN point. The calibrated parameters are listed in \tabref{tab:beta_values}.

\subsection{Power consumption and mode crossover analysis}
\label{subsec:power_analysis}

\begin{figure}[t]
\centering
\includegraphics[width=0.95\columnwidth]{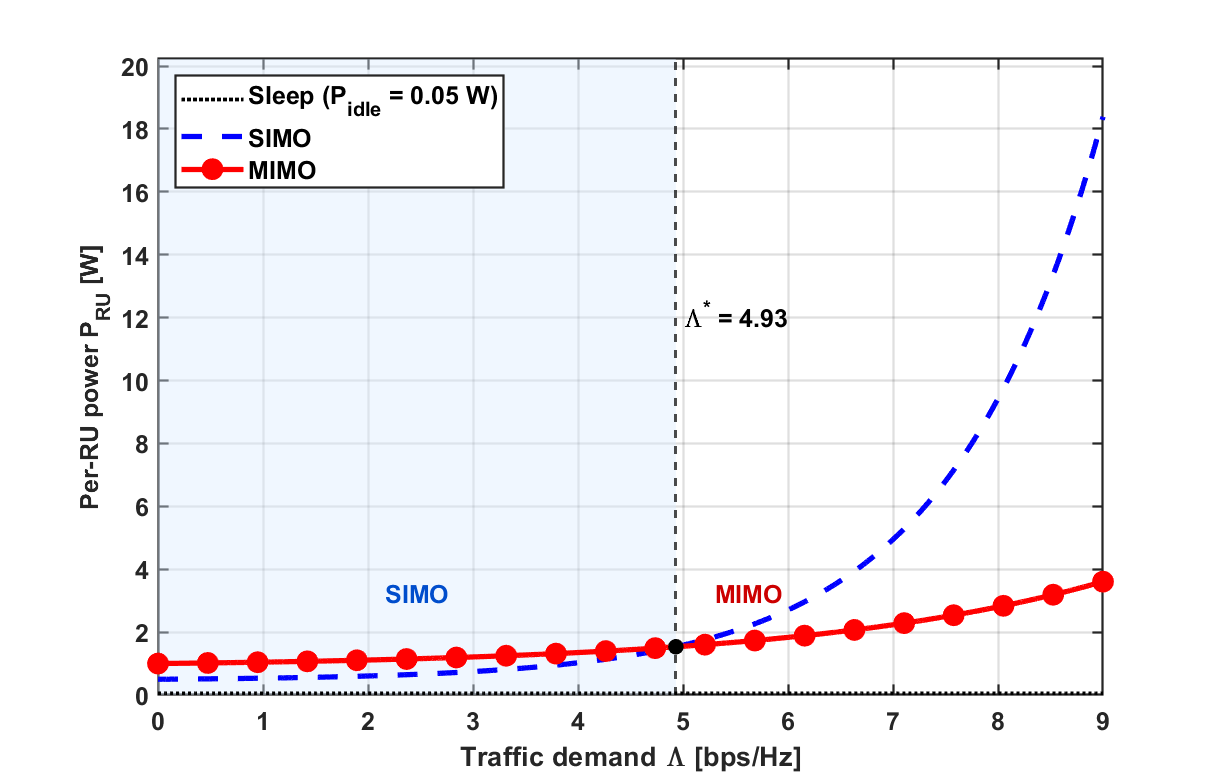}
\caption{Per-RU power consumption versus traffic demand $\Lambda$ for MIMO, SIMO, and sleep modes under an RF-only model~\cite{KimEtAlCrosslayerApproachEnergy2009} with $\Nt = 2$ antennas per RU. The MIMO/SIMO crossover occurs at $\Lambda^* = 4.93$~bps/Hz, below which SIMO is more power-efficient due to reduced RF chain overhead.}
\label{fig:fig6}
\end{figure}

\begin{figure}[t]
\centering
\includegraphics[width=0.95\columnwidth]{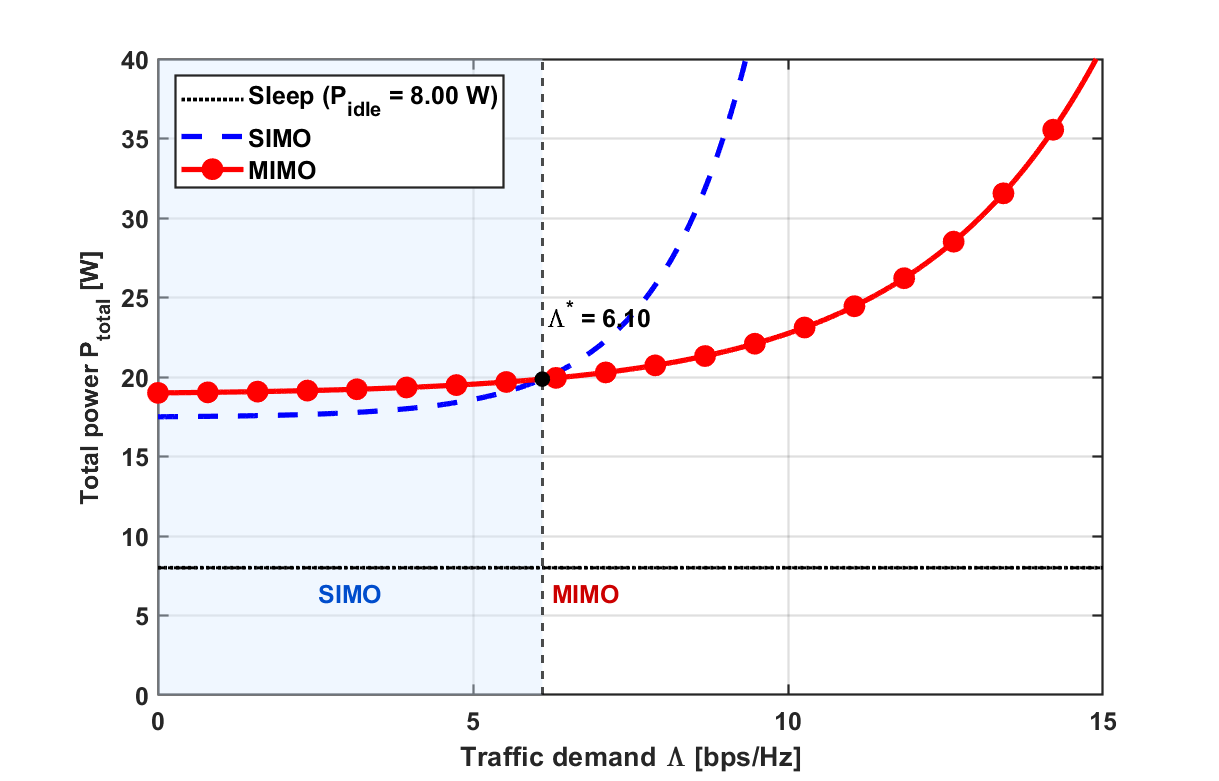}
\caption{Total system power consumption versus traffic demand $\Lambda$ under the proposed vRAN-aware model including vDU baseband and fronthaul transport. The crossover shifts to $\Lambda^* = 6.10$~bps/Hz, a 24 \% increase over the RF-only estimate in \figref{fig:fig6}, extending the SIMO power-efficient range.}
\label{fig:fig7}
\end{figure}

\begin{table}[t]
\caption{Calibrated EESM parameters $\beta_m$ for MCS indices $m \in \{0, 1, 2, 3\}$ over TDL-C channels.}
\label{tab:beta_values}
\centering
\small
\begin{tabular}{cccc}
\hline
$\beta_0$ & $\beta_1$ & $\beta_2$ & $\beta_3$ \\
\hline
0.5213 & 2.2514 & 14.7577 & 114.7487 \\
\hline
\end{tabular}
\end{table}

We analyze the per-RU and total system power consumption as a function of traffic demand $\Lambda$ to determine the MIMO/SIMO crossover threshold $\Lambda^*$. Prior RF-only models compute $\Lambda^*$ based on RU transmit and static power alone~\cite{KimEtAlCrosslayerApproachEnergy2009}, while the proposed vRAN-aware model additionally accounts for vDU baseband processing and fronthaul transport.

Figure~\ref{fig:fig6} shows the per-RU power under RF-only modeling. The MIMO and SIMO curves cross at $\Lambda^* = 4.93$~bps/Hz. Below this threshold, SIMO consumes less power due to its lower static RF chain overhead $P_{\text{RU},0}^{(\text{SIMO})}$. Above it, MIMO becomes more power-efficient by spreading traffic demand across two spatial streams, reducing the per-stream transmit power requirement.

Figure~\ref{fig:fig7} shows the total system power including vDU and fronthaul. The crossover shifts to $\Lambda^* = 6.10$~bps/Hz, a 24\% increase over the RF-only estimate, due to additional vDU dynamic power and the fronthaul overhead incurred when switching to MIMO.

\subsection{Energy efficiency performance}
\label{subsec:ee_results}

We evaluate the EE performance of the proposed framework. We first examine the per-MCS EE under static channel conditions, then analyze the daily EE profile under MATSim-generated traffic. The evaluation applies EESM link abstraction to 100 TDL-C realizations per operating point.

\subsubsection{Without traffic constraints}
\label{subsubsec:ee_no_traffic}

Figure~\ref{fig:fig8_dmimo_ee_mcs} shows the per-MCS EE as a function of TDL operating SNR under the optimization framework in Section~\ref{sec:EE_DMIMO}. The step-function approximation in \eqref{eq:bler_step} sets EE to zero while the BLER stays above the target $p_e^{\text{target}} = 10^{-2}$, and it raises EE to a peak once the BLER reaches the target. This transition marks successful decoding, not the onset of system operation. The BLER reaches the target at $\gamma_{\text{th}}^{(m)}$, an AWGN reference. Over the TDL-C channel, frequency-selective fading adds a margin, so the target BLER is met at a higher channel SNR. QPSK reaches the target near $10$~dB, above its $2.9$~dB AWGN threshold. The EE then decays as the transmit power in the denominator of \eqref{eq:ee_metric} grows faster than the fixed numerator $\NTBS / T_{\text{slot}}$. Unlike classical waterfilling, which raises capacity monotonically with transmit power, the EE-optimal allocation fixes a unique optimal total power through $\lambda_m$ in \eqref{eq:lambda_closed}. Algorithm~\ref{alg:joint} selects the MCS that maximizes EE at each SNR point, tracing the upper envelope of the four MCS curves. For reference, we plot the idealized Shannon-capacity EE by replacing the EESM throughput with $\NRE^{\text{PRB}} n r \log_2(1 + \text{SNR}_{\text{EESM}}) / T_{\text{slot}}$ under the same power model. The idealized curve exceeds the EESM envelope by $1.4$--$2.0\times$ over the MCS-3 peak operating range ($32$--$45$~dB). It remains positive at MCS threshold gaps where the EESM throughput drops to zero. This gap quantifies how far the practical MCS-discrete EE operates below the continuous-rate ideal.

\begin{figure}[t]
\centering
\includegraphics[width=0.95\columnwidth]{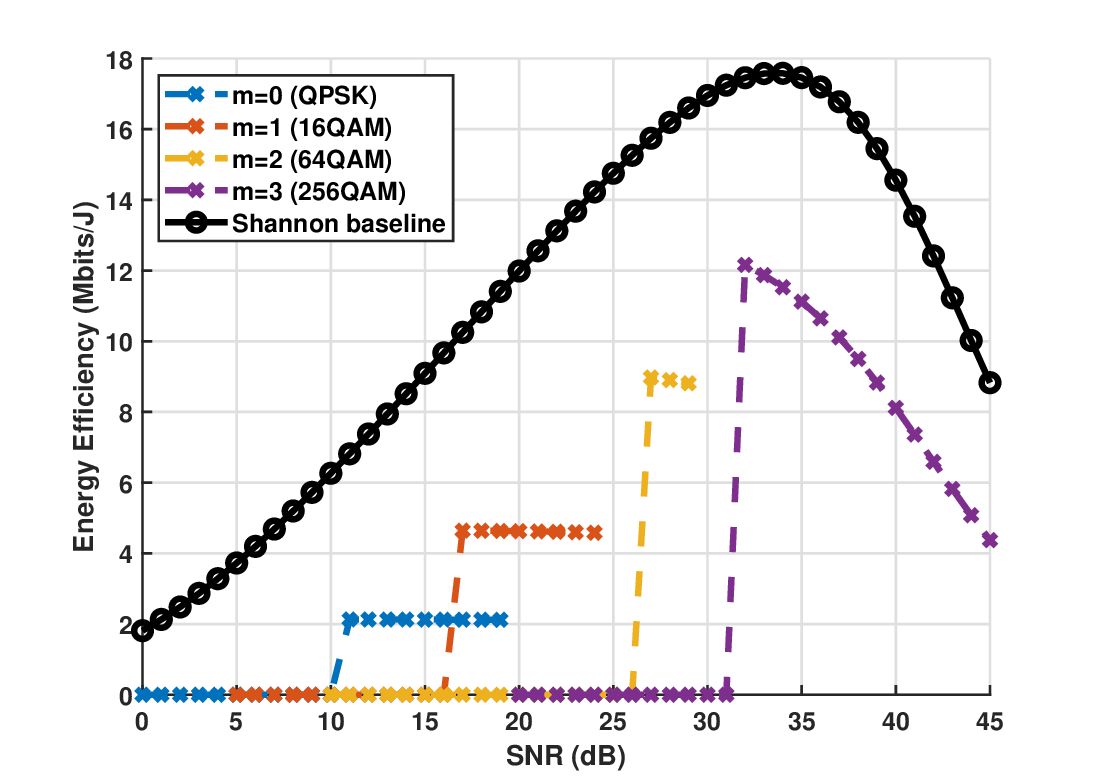}
\caption{Energy efficiency versus TDL-C channel SNR for all MCS with $N_{\text{RB}} = 273$. Each curve jumps to a peak once the BLER reaches the target $10^{-2}$, which the TDL-C channel meets above the AWGN threshold by a fading margin. Higher-order MCS reaches greater peak EE at a higher required SNR. The Shannon-capacity curve marks the idealized continuous-rate upper bound, which the MCS-discrete EESM envelope approaches under practical 3GPP constraints.}
\label{fig:fig8_dmimo_ee_mcs}
\end{figure}

\begin{remark}
A higher MCS raises both the rate $R_m$ and the threshold $\gamma_{\text{th}}^{(m)}$, so each MCS is EE-optimal only in a narrow SNR window above its threshold, and the EE-optimal MCS shifts with the operating SNR.
\end{remark}

\subsubsection{Daily traffic-aware optimization}
\label{subsubsec:ee_traffic}

Figure~\ref{fig:hexatraffic_3d} shows the 30 day average traffic demand $\Lambda(t)$ across all RUs at four representative time slots: 00:00, 07:00, 14:00, and 21:00, selected from the full 24 hour simulation. The outer loop transitions each RU among MIMO, SIMO, and sleep modes based on $\Lambda_\ell(t)$ relative to the crossover threshold $\Lambda^*$ and sleep threshold $\Lambda^{\text{sleep}}$ derived from \eqref{eq:crossover}. 

\begin{figure*}[t]
    \centering
    \includegraphics[width=0.95\textwidth]{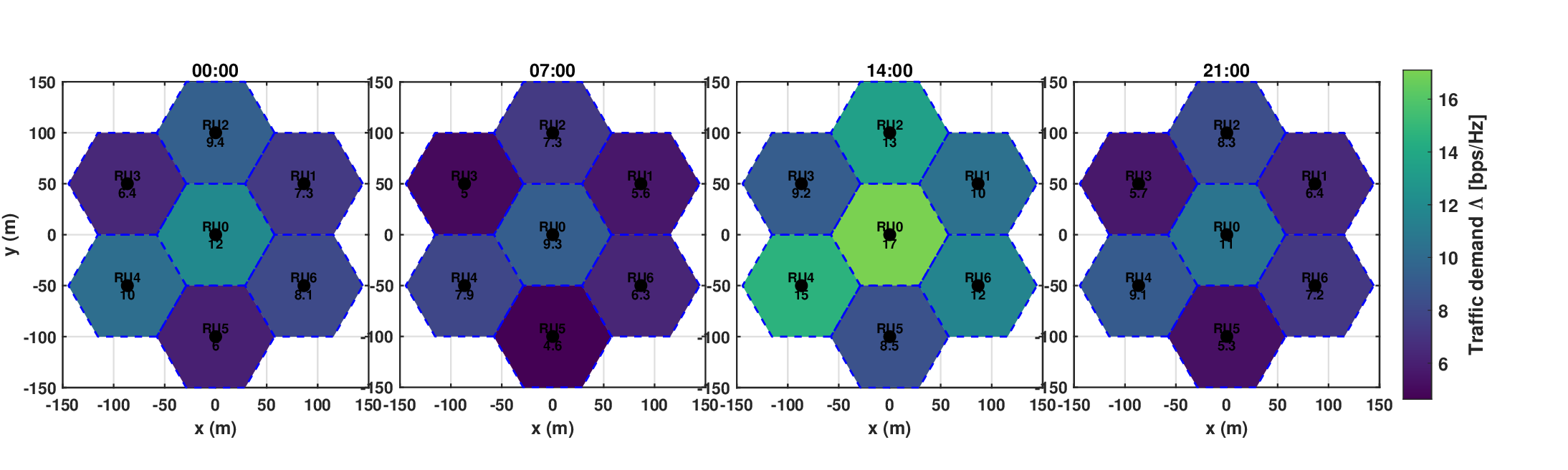}
    \caption{Thirty-day average aggregate traffic demand per RU coverage area for $L = 7$ RUs in a UMi layout with $d_{\text{ISD}} = 100$ m, at four representative time slots from a 24 hour MATSim simulation under the APN-C traffic profile. Color intensity indicates aggregate traffic load in bps per cell. Traffic concentrates at central RUs during daytime hours and drops at night, confirming the diurnal spatio-temporal variation that motivates adaptive RU mode selection.}
    \label{fig:hexatraffic_3d}
\end{figure*}

Figure~\ref{fig:daily_ee} shows the vRAN EE over a 24 hour period for the proposed traffic-aware framework with $L_{\text{act}} = 3$ and three traffic-agnostic always-on MIMO baselines with $L_{\text{act}} \in \{3, 5, 7\}$. The vDU updates the mode, MCS, and power allocation at each 15 minute scheduling slot, yielding 96 slots over the full day. The four configurations share the same EESM-based link abstraction, power model, and MCS selection, differing only in RU mode control The $L_{\text{act}} = 5$ and $L_{\text{act}} = 7$ baselines nearly match the proposed framework at peak hours as the coherent combining gain saturates, but fall below the proposed framework during off-peak hours when static MIMO power dominates. The $L_{\text{act}} = 3$ baseline incurs the largest PA power penalty since fewer RUs must each transmit at higher power to satisfy the EESM threshold. The proposed framework achieves 23~Mbits/J, or 32\%, average EE gain over the $L_{\text{act}} = 3$ baseline by transitioning RUs into SIMO and sleep modes when $\Lambda(t)$ falls below $\Lambda^*$. This result demonstrates that traffic-aware mode adaptation yields larger EE gain than simply activating more RUs.

\begin{figure}[tb]
\centering
\includegraphics[width=1.05\columnwidth]{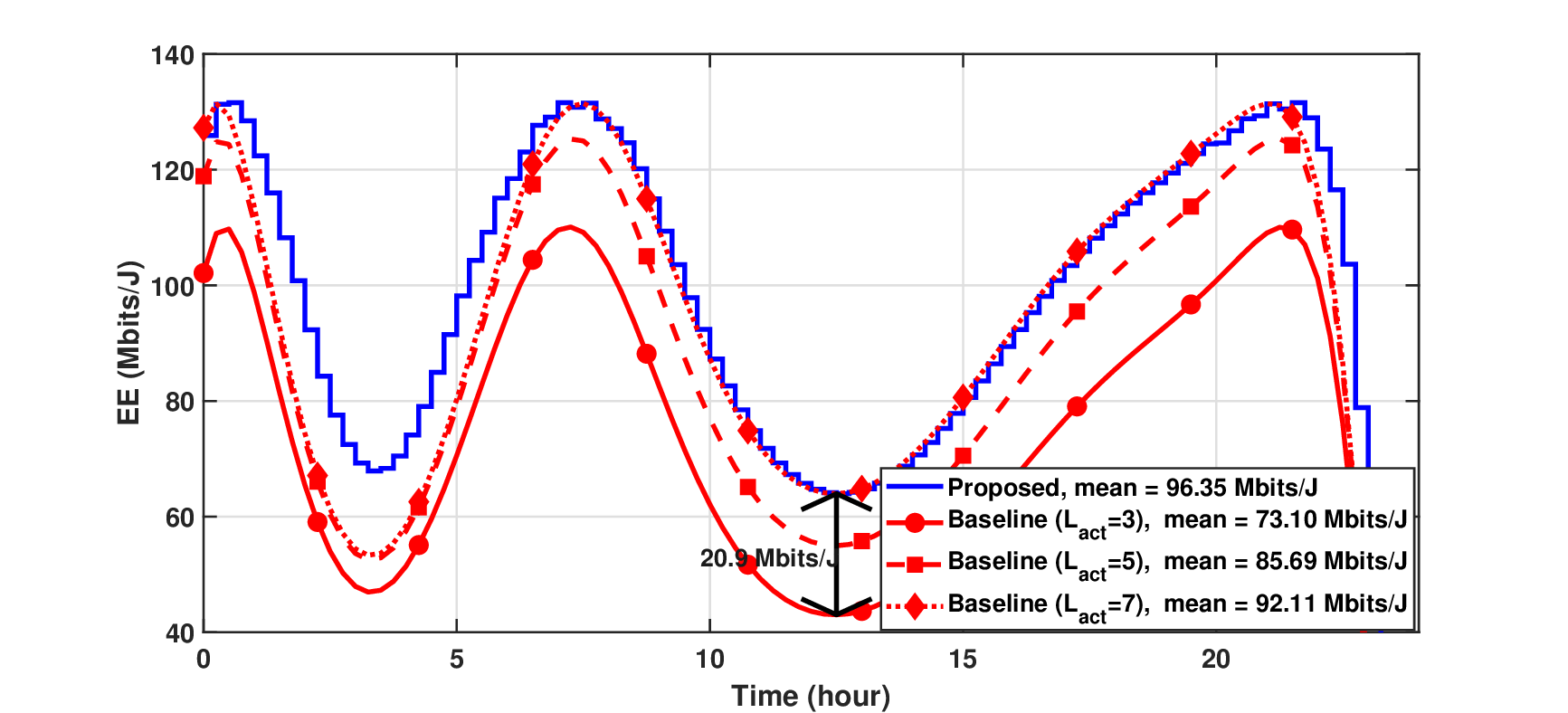}
\caption{Per 15-minute slot vRAN energy efficiency over a 24 hour period with $\Nt = 2$ antennas per RU. The proposed traffic-aware framework with $L_{\text{act}} = 3$ achieves 23~Mbits/J, or 32\%, average EE gain over the $L_{\text{act}} = 3$ always-on MIMO baseline. The proposed framework outperforms always-on baselines at $L_{\text{act}} \in \{3, 5, 7\}$ throughout the day by transitioning RUs among MIMO, SIMO, and sleep modes according to $\Lambda^*$.}
\label{fig:daily_ee}
\vspace{-5pt}
\end{figure}

\section{Conclusion and future work}
\label{sec:conclusion}

In this paper, we proposed a cross-layer energy efficiency (EE) optimization framework for distributed MIMO (D-MIMO) vRAN systems. The framework reformulated EE maximization as a convex power minimization problem with a closed-form, waterfilling-like solution. This framework is not tied to 3GPP 5G NR. It depends only on the finite cardinality of the MCS, PRB, and rank sets, so it remains relevant for 6G and beyond once their parameter values are specified. The vRAN-aware power model revealed that vDU and fronthaul overhead shifted the MIMO/SIMO crossover threshold by 24\% beyond the RF-only estimate. This shift extended the traffic range over which SIMO operation was more energy-efficient. The traffic-aware mode control achieved a 32\% average EE gain over a traffic-agnostic, always-on MIMO baseline at $L_{\text{act}} = 3$. The framework achieved a 12\% average gain over the $L_{\text{act}} = 5$ baseline and 5\% over the $L_{\text{act}} = 7$ baseline, concentrated in off-peak hours when SIMO and sleep modes dominated. These baselines closely matched the framework at peak hours. 

Our work assumed perfect CSI via TDD reciprocity and perfect synchronization across distributed RUs. In practice, RU synchronization offsets, power-dependent PA efficiency, and correlated fading across nearby RUs sharing scatterers would degrade the effective SNR and shift the optimal MCS and power allocation. The reported EE gains should be interpreted as an upper bound. The single-user assumption also limited applicability to multi-user deployments. In multi-user scenarios, channel estimation errors at the transmitter cause inter-user interference that per-user receivers cannot fully cancel. Future work will extend the framework to such scenarios with joint precoder design and per-user power allocation. Incorporating per-stream waterfilling across unequal singular values may further improve MIMO-mode EE, albeit at the cost of increased optimization complexity. Finally, the Open RAN near-real-time RIC offers a natural platform for coordinating RU mode control across cells. Digital twins complement the analytical EESM framework with high-fidelity simulation. Combining closed-form optimization with digital-twin verification can bridge analytical tractability and deployment realism.

\appendices

\section{Proof of Theorem~\ref{thm:optimal_power}}
\label{app:proof_theorem1}

We set $r = 1$ in the multi-stream Lagrangian and use $a_k \equiv a_{1,k} = \sigma_1^2[k]/(\beta N_0 \Delta f)$. The KKT conditions for problem \eqref{eq:power_min} consist of primal feasibility, dual feasibility $\lambda, \nu, \xi_k \geq 0$, complementary slackness, and the stationarity condition obtained by specializing \eqref{eq:kkt_stationarity} to $r = 1$,
\begin{equation} 
    1 + \nu - \frac{\lambda a_k}{K+\Ncp} e^{-a_k P[k]} - \xi_k = 0. 
\label{eq:kkt_stationarity_ss} 
\end{equation}

We first show $\lambda > 0$. If $\lambda = 0$, stationarity forces $\xi_k = 1 > 0$ for all $k$, so complementary slackness sets $P^*[k] = 0$. This violates the EESM constraint, since $K/(K+\Ncp) > \exp(-\gamma_{\text{th}}^{(m)}/\beta_m)$ for all MCS indices. The objective minimizes $\mathbf{1}^T \mathbf{p}$, so the optimum meets the EESM constraint with the smallest power. Whenever the total power stays below $P_{\max}$, complementary slackness sets $\nu = 0$.

For an active subcarrier with $P^*[k] > 0$, we have $\xi_k = 0$. The stationarity condition then gives
\begin{equation}
e^{-a_k P^*[k]} = \frac{K+\Ncp}{\lambda a_k}.
\label{eq:proof_step2}
\end{equation}
Taking the logarithm and substituting $a_k$ yields \eqref{eq:optimal_power}. Since $\lambda > 0$, the EESM constraint holds with equality. Substituting the active-subcarrier solution and $P^*[k] = 0$ on the inactive set $\cI$ into the constraint and solving for $\lambda$ yields \eqref{eq:lambda_closed}. A subcarrier is inactive when $\lambda a_k/(K+\Ncp) \leq 1$, which makes the logarithm in \eqref{eq:optimal_power} nonpositive. \hfill$\blacksquare$


\bibliographystyle{IEEEtran}
\bibliography{JP_EE}

\end{document}